\documentclass[preprint,3p,times]{elsarticle}
\RequirePackage{multirow,booktabs,subfigure,color,array,hhline,makecell}
\usepackage{amssymb}
\usepackage{amsmath}
\usepackage{graphicx}
\usepackage{subfigure}
\usepackage{amsthm}
\usepackage{bbm}
\usepackage{mathrsfs}
\usepackage{indentfirst}
\usepackage[colorlinks,citecolor=blue,urlcolor=blue]{hyperref}
\allowdisplaybreaks
\usepackage[table]{xcolor}
\usepackage{tikz-network}
\usepackage{pgf}
\usepackage{tikz}
\usepackage{subfigure}
\usepackage{matlab-prettifier}
\usepackage{tabularx}
\usepackage{ragged2e}
\newcolumntype{Y}{>{\RaggedRight\arraybackslash}X}

\usetikzlibrary{arrows, decorations.pathmorphing, backgrounds, positioning, fit, petri, automata}
\usetikzlibrary{shadows,arrows,positioning}
\RequirePackage{algorithm,algpseudocode}
\allowdisplaybreaks
\usepackage{changes}
\usepackage{caption}
\usepackage{bbm}
\usepackage{setspace}
\newtheorem{thm}{Theorem}
\newtheorem{defin}{Definition}
\newtheorem{lem}{Lemma}
\newtheorem{assum}{Assumption}
\newtheorem{rem}{Remark}

\newtheorem{Ex}{Example}
\allowdisplaybreaks[4]
\AtBeginDocument{%
    \providecommand\BibTeX{%
        \normalfont B\kern-0.5em{\scshape i\kern-0.25em b}\kern-0.8em\TeX}}
\journal{~}
\begin{document}
\captionsetup[figure]{labelfont={bf},labelformat={default},labelsep=period,name={Fig.}}
\begin{frontmatter}
\title{Joint estimation of asymmetric community numbers in directed networks}
\author[label1]{Huan Qing\corref{cor1}}
\ead{qinghuan@cqut.edu.cn$\&$qinghuan@u.nus.edu}
\address[label1]{School of Economics and Finance, Chongqing University of Technology, Chongqing, 400054, China}
\cortext[cor1]{Corresponding author.}
\begin{abstract}
Community detection in directed networks is a central task in network analysis. Unlike undirected networks, directed networks encode inherently asymmetric relationships, giving rise to sender and receiver roles that may each follow distinct community organizations with possibly different numbers of communities. Estimating these two community counts simultaneously is therefore considerably more challenging than in the undirected setting, yet it is essential for faithful model specification and reliable downstream inference. This work addresses this challenge within the stochastic co-block model (ScBM), a powerful statistical framework for capturing asymmetric relational structures inherent in directed networks. We propose a novel goodness-of-fit test based on the deviation of the largest singular value of a normalized residual matrix from the constant value 2. We show that the upper bound of this test statistic converges to zero under the null hypothesis, while this statistic goes to infinity if the true model has finer communities than hypothesized. Leveraging this tail bounds behavior, we develop an efficient sequential testing algorithm that lexicographically explores candidate community number pairs. To enhance robustness in practical settings, we further introduce a ratio-based variant that detects the transition point in the test statistic sequence. We rigorously show both algorithms' consistency in recovering the true sender and receiver community counts under ScBM. Numerical experiments demonstrate the accuracy and robustness of our methods in estimating community numbers across diverse ScBM settings. 
\end{abstract}
\end{frontmatter}
\section{Introduction}\label{sec1}
Directed networks are ubiquitous in real-world complex systems, capturing asymmetric relationships where interactions exhibit inherent directionality \citep{malliaros2013clustering}. For example, in social media platforms like Twitter, users follow one another directionally, creating a directed network of influence. In email exchanges, senders send emails to receivers, creating send-to relationships \citep{rohe2016co}. In academic citation networks, papers cite others without reciprocal acknowledgment, forming a directed structure of knowledge flow \citep{PengshengAOAS896,ji2022co}. These examples illustrate that directed networks encode edge relationships with distinct ``sender" (edge origins) and ``receiver" (edge destinations) roles, fundamentally different from undirected networks where connections are symmetric \citep{rohe2016co}. The study of such asymmetric relational patterns is critical for understanding complex systems, as the directionality often encodes functional or causal dependencies.

Community detection in directed networks aims to identify groups of nodes with similar interaction patterns, distinguishing between sender-based and receiver-based communities \citep{leicht2008community,malliaros2013clustering,rohe2016co}. This task is crucial for applications such as identifying opinion leaders in social networks, mapping intellectual communities in citation graphs, optimizing traffic routing in transportation systems, and market segments in stock networks \citep{li2022undirected}. Unlike undirected networks, where communities imply symmetric associations, directionality of directed networks introduces unique challenges: sender communities (nodes with similar outgoing edge patterns) and receiver communities (nodes with similar incoming edge patterns) may exhibit distinct organizational structures, including potentially different numbers of communities \citep{rohe2016co}. Faithfully recovering these asymmetric structures is essential for learning complex structures in directed networks.

The stochastic co-block model (ScBM) proposed in \citep{rohe2016co} is a powerful statistical framework for modeling directed networks with asymmetric community structure. It extends the classical stochastic block model (SBM) \citep{holland1983stochastic} from undirected networks to directed networks by introducing dual community memberships: each node belongs to a sender community governing its outgoing edges and to a receiver community influencing its incoming edges. Formally, ScBM assumes that each node belongs to a sender community and a receiver community, and the probability of an edge from node $i$ to node $j$ depends on the sender community of $i$ and the receiver community of $j$. This dual assignment mechanism enables ScBM to capture the inherent directionality of relationships while accommodating potential asymmetry between sending and receiving behaviors. ScBM is particularly powerful for capturing scenarios where the number of sender communities differs from that of receiver communities. Numerous community detection methods have been developed for ScBM, including various spectral clustering approaches \citep{rohe2016co,zhou2019analysis,wang2020spectral,qing2023community,guo2023randomized} and variational inference algorithms \citep{zhou2020optimal,zhang2022identifiability}. However, all these existing methods require that both sender community count \(K_s\) and receiver community count \(K_r\) are known in advance, which is a significant practical limitation since the numbers of sender and receiver communities in real-world directed networks are rarely known. This limitation necessitates reliable joint estimation of both community numbers before meaningful community detection can proceed.

In undirected networks, community number estimation under SBM or its degree-corrected extension \citep{karrer2011stochastic} has been addressed through various techniques including likelihood ratio tests \citep{LikeAos2017,ma2021determining}, Bayesian information criterion \citep{mcdaid2013improved,saldana2017many,hu2020corrected},  network cross-validation \citep{chen2018network,li2020network}, spectrum of the Bethe Hessian matrices \citep{CanMK2022,hwang2024estimation}, and goodness-of-fit tests \citep{bickel2016hypothesis,lei2016goodness,dong2020spectral,hu2021using,jin2023optimal,wu2024spectral}. 
However, these approaches fail in directed settings due to the dual community assignment mechanism of ScBM, where the number of sender communities may even be different from that of receiver communities. To our knowledge, no existing method provides theoretically guaranteed joint estimation of $(K_s, K_r)$ for directed networks under ScBM. This paper bridges the gap by developing a framework for jointly estimating the numbers of sender and receiver communities within the ScBM framework. Our key contributions are threefold:
\begin{itemize}
    \item We propose a novel test statistic based on the deviation of the largest singular value of a normalized residual matrix from the constant 2. This test statistic leverages the tail behavior of singular values rather than their asymptotic distribution—a crucial innovation necessitated by the asymmetric nature of directed networks. 
    
    \item We establish rigorous theoretical guarantees to show that under correct model specification, the test statistic's upper bound converges to zero with high probability, while under underfitted models, the statistic diverges to infinity. This dichotomous behavior provides a statistical foundation for model selection in directed networks.
    
    \item We propose two computationally efficient algorithms: a sequential testing procedure that lexicographically explores candidate asymmetric community number pairs using a decaying threshold, and a ratio-based variant that detects the transition point in the test statistic sequence for enhanced robustness in practical settings. We prove the consistency of both estimators under ScBM with mild conditions and conduct extensive numerical experiments to validate their accuracy across diverse ScBM configurations, demonstrating robustness to threshold selection and applicability to real-world directed networks.
\end{itemize}
\section{Model formulation and problem statement}\label{sec:model}
This section formally defines the Stochastic co-block model (ScBM) and precisely states the core estimation problem. We specify the ScBM's generative mechanism, which incorporates distinct sender and receiver community memberships to model asymmetric edge directionality. The central problem addressed is the joint estimation of the unknown sender and receiver community counts \((K_s, K_r)\) from observed network data. Necessary theoretical assumptions for establishing theoretical guarantees are also presented.
\subsection{Stochastic co-block model (ScBM)}
In this paper, we study a directed network on $n$ vertices represented by an adjacency matrix $A\in\{0,1\}^{n\times n}$, where $A_{ii}=0$ and $A_{ij}=1$ signifies a directed edge sending from node $i$ to node $j$.  
The \emph{stochastic co-block model} (ScBM) partitions the $n$ vertices into $K_s$ \emph{sender communities} and $K_r$ \emph{receiver communities} \citep{rohe2016co}. Formally, the model is specified as follows:
\begin{defin}[Stochastic co-block model (ScBM)]\label{def:ScBM}
Consider a directed network with $n$ nodes. Let $A \in \{0,1\}^{n \times n}$ be the adjacency matrix where $A_{ii} = 0$ and $A_{ij} = 1$ indicates a directed edge from node $i$ to node $j$. The stochastic co-block model (ScBM) is parameterized by:
\begin{itemize}
    \item Sender community labels: $g^s \in \{1,\dots,K_s\}^n$.
    \item Receiver community labels: $g^r \in \{1,\dots,K_r\}^n$.
    \item Block probability matrix: $B \in [0,1]^{K_s \times K_r}$.
\end{itemize}
Given $(g^s, g^r, B)$, the entries of $A$ are independent with
\[
\mathbb{P}(A(i,j)= 1) = B(g^s(i), g^r(j)), \quad \forall i \neq j.
\]
The expected adjacency matrix $\Omega$ satisfies $\Omega(i,j) = B(g^s(i), g^r(j))$ for $i \neq j$ and $\Omega(i,i)= 0$.
\end{defin}
Sure, the sender-receiver asymmetry allows $g^s \neq g^r$ and $K_s \neq K_r$, enabling flexible modeling of asymmetric relational patterns in directed networks.
\subsection{Core problem: joint community number estimation}
In this paper, we address the fundamental problem of \emph{community number estimation} in ScBMs: given a directed network, how can we jointly determine the numbers of sender communities $K_s$ and receiver communities $K_r$? We formulate this as a sequential goodness-of-fit testing problem: for candidate pairs $(K_{s0}, K_{r0})$, we test
\[
H_0: (K_s, K_r) = (K_{s0}, K_{r0}) \quad \text{versus} \quad H_1: K_s > K_{s0} \text{ or } K_r > K_{r0},
\]
where $H_1$ explicitly encodes \emph{underfitting} scenarios where the hypothesized model lacks sufficient sender or receiver communities. This formulation is statistically challenging and practically significant for three reasons:
\begin{enumerate}
    \item The bivariate community structure $(K_s, K_r)$ creates a combinatorial search space of size $(K_s+K_r)^2$, requiring efficient testing strategies.
    \item Underfitting (failing to reject $H_0$ when $H_1$ holds) leads to loss of structural resolution, while overfitting is mitigated algorithmically via ordered search.
    \item Existing goodness-of-fit tests for undirected networks rely on symmetric eigenvalue distributions (e.g., Tracy-Widom laws analyzed in \citep{lei2016goodness}) that do not extend to the asymmetric singular value decompositions required for directed networks.
\end{enumerate}

Given that few existing methods provide theoretically guaranteed joint estimation of $(K_s, K_r)$ in ScBM for directed networks, this work bridges this gap by developing a testing framework based on singular value tail bounds of normalized residual matrices. Our approach leverages the asymptotic behavior of the largest singular value under $H_0$ versus $H_1$, enabling consistent community number estimation.
\subsection{Theoretical assumptions}\label{subsec:assumptions}
To establish theoretical guarantees, we require the following regularity conditions:

\begin{assum}[Bounded edge probabilities]\label{assump:a1}
The block probability matrix satisfies $\delta \leq B(k,l) \leq 1 - \delta$ for all $k,l$ and some $\delta > 0$. 
\end{assum}
Assumption \ref{assump:a1} ensures well-defined variance in the residual matrix defined later and excludes degenerate cases where normalization becomes unstable. 

\begin{assum}[Balanced community sizes]\label{assump:a2}
The community sizes satisfy:
\[
\min_{k=1,\dots,K_s} |\{i: g^s(i) = k\}| \geq c_0 \frac{n}{K_s}, \quad 
\min_{l=1,\dots,K_r} |\{j: g^r(j) = l\}| \geq c_0 \frac{n}{K_r}
\]
for some $c_0 > 0$. 
\end{assum}
Assumption \ref{assump:a2} prevents any community from being asymptotically negligible, guaranteeing that the sample size within every block grows linearly with $n$. Similar assumptions are also required for theoretical guarantees of goodness-of-fit test in undirected networks under SBM \citep{lei2016goodness,hu2021using,wu2024spectral,wu2024two}

Define $K_{\max} = \max(K_s,K_r)$ and the community separation metric:
\[
\delta_n = \min\left( \min_{k \neq k'} \max_{l} |B(k,l) - B(k',l)|, \min_{l \neq l'} \max_{k} |B(k,l) - B(k,l')| \right).
\]

The community separation $\delta_n$ quantifies the minimal distinguishability between communities through capturing the worst-case discriminability between any two sender or receiver communities.
\begin{assum}[Model complexity and community distinguishability]\label{assump:a3}
\[
\frac{K^{2}_{\max}\max(\log n, \delta^{-2}_{n})}{n} \to 0 \quad \text{as} \quad n \to \infty,
\]
which also implies: (1) $K_{\max} \ll \sqrt{n/\log n}$ limiting model complexity; (2) $\delta_n \gg n^{-1/2}$ ensuring communities are statistically distinguishable.
\end{assum}
Assumption \ref{assump:a3} establishes a fundamental trade-off between model complexity captured by $K_{\max}$ and community separability described by $\delta_n$ that ensures the asymptotic controllability of the test statistic. Further understandings of Assumption \ref{assump:a3} are provided in the following remark.
\begin{rem}\label{rem:assump3_interpretation}
The condition $\frac{K^{2}_{\max}\max(\log n, \delta^{-2}_{n})}{n} \to 0$ adaptively constrains the interplay between model complexity and community separability depending on the relative scaling of $\delta_n$ and $n$, which we analyze through two typical regimes:

Case 1: $\delta^{-2}_{n} \geq \log n$ (weak separation). Here $\max(\log n, \delta^{-2}_{n}) = \delta^{-2}_{n}$, reducing the assumption to $\frac{K^{2}_{\max}}{n\delta^{2}_{n}} \to 0$. This yields the sandwich bound
    \[
    \frac{K_{\max}}{\sqrt{n}} \ll \delta_n \leq \frac{1}{\sqrt{\log n}}.
    \]
    
The sandwich bound explicitly balances separation and complexity: models with larger $K_{\max}$ require stronger separation (larger $\delta_n$), while weaker community separation (smaller $\delta_n$) requires fewer communities. Meanwhile, he minimal separation must dominate the parametric rate $K_{\max}/\sqrt{n}$. For fixed $K_{\max}$, $\delta_n$ must decay slower than $n^{-1/2}$. If $K_{\max} \to \infty$ slowly (e.g., $K_{\max} = O(\log n)$), $\delta_n$ must decay slower than $\log n/\sqrt{n}$.

Case 2: $\delta^{-2}_{n} < \log n$ (strong separation). Here $\max(\log n, \delta^{-2}_{n}) = \log n$, simplifying the assumption to $\frac{K^{2}_{\max} \log n}{n} \to 0$, equivalent to $K_{\max} \ll \sqrt{n / \log n}$. In this regime, separation is relatively strong ($\delta_n > \frac{1}{\sqrt{\log n}}$), so community distinguishability is not the binding constraint. The assumption reduces to limiting model complexity: $K_{\max}$ must grow slower than $\sqrt{n / \log n}$. For fixed $K_{\max}$, the condition holds trivially as $n \to \infty$. 
\end{rem}
\section{Goodness-of-fit test}\label{sec:test}
This section develops a theoretically grounded goodness-of-fit test for ScBM. We first introduce an ideal test statistic using oracle parameters, then derive its practical counterpart with estimated parameters, and finally establish its asymptotic behavior under both null and alternative hypotheses. The core innovation lies in leveraging singular value tail bounds of normalized residual matrices to detect community underfitting.
\subsection{Ideal test statistic and its asymptotic behavior}
To formalize the test, we begin with the \emph{ideal residual matrix} \(R\) constructed using true parameters:
\[
R(i,j) = 
\begin{cases} 
\dfrac{A(i,j) - \Omega(i,j)}{\sqrt{(n-1) \Omega(i,j) (1 - \Omega(i,j))}} & i \neq j, \\
0 & i = j,
\end{cases}
\]
where \(\Omega(i,j) = B(g^s(i), g^r(j))\) is the true edge probability. This normalization ensures \(\mathbb{E}[R(i,j)] = 0\) and \(\operatorname{Var}(R(i,j)) = \frac{1}{n-1}\) for \(i \neq j\), transforming \(R\) into a generalized random non-symmetric matrix with controlled variance. The ideal test statistic is defined as:
\begin{align}\label{idealTestStatistic}
T_n = \sigma_1(R) - 2,
\end{align}
where \(\sigma_1(\cdot)\) denotes the largest singular value. The shift by 2 anticipates the asymptotic behavior of \(\sigma_1(R)\) under \(H_0\), as established below.

\begin{lem}\label{ideal0}
Under \(H_0\) and Assumption \ref{assump:a1}, for any \(\epsilon > 0\), we have
\[
\mathbb{P}(T_n < \epsilon) \to 1 \quad \text{as} \quad n \to \infty.
\]
\end{lem}
\begin{proof}[Proof intuition]
The convergence follows from tail bounds for the largest singular value of random matrices with independent entries \citep{Afonso2016}. Specifically, Lemma \ref{Extension} in the Appendix yields \(\|R\| \leq 2(1+\eta) + \mathcal{O}_P\big(\sqrt{\log n / n}\big)\) for any \(\eta > 0\). The boundedness of \(\Omega(i,j)\) guaranteed by Assumption \ref{assump:a1} controls $R$'s variance.
\end{proof}

\begin{rem}
A crucial distinction arises when comparing to undirected networks: while \citep{lei2016goodness} leveraged Tracy-Widom distributions for symmetric adjacency matrices, asymptotic distributions for directed residuals remain an open problem. This theoretical gap originates from the inherent \textit{asymmetry} of the ideal residual matrix $R$ in directed networks. Therefore, to rigorously determine the number of sender and receiver communities under the ScBM, we focus on characterizing the asymptotic behavior of our test statistic, rather than its asymptotic distribution.
\end{rem}
\subsection{Test statistic and its asymptotic behavior}
In the construction of the ideal test statistic $T_{n}$, we have used unknown model parameters $(B, g^{s}, g^{r})$ for real directed networks. Thus, $T_{n}$ can not be used as a practical test statistic for real data. Next, we introduce an approximation of $T_{n}$ by estimating the block probability matrix $B$ and community labels $g^{s}$ and $g^{r}$ from the adjacency matrix $A$ of a real directed network.

Let \(\mathcal{M}\) be any community detection algorithm for directed networks. Apply \(\mathcal{M}\) to \(A\) to partition the \(n\) nodes into \(K_{s0}\) \textit{sender communities} and \(K_{r0}\) \textit{receiver communities}, yielding estimated labels \(\hat{g}^s\) for sender community and \(\hat{g}^r\) for receiver community. Then, we compute $B$'s plug-in estimator    \(\hat{B} \in [0,1]^{K_{s0} \times K_{r0}}\) via  
\[  
\hat{B}(k,l) = \frac{\sum_{i:\hat{g}^s(i) = k} \sum_{j:\hat{g}^r(j) = l} A(i,j)}{|\{i: \hat{g}^s(i) = k\}| \cdot |\{j: \hat{g}^r(j) = l\}|}, \quad k=1,\dots,K_{s0}, \, l=1,\dots,K_{r0}.  
\]  

From $\hat{B}$'s construction, we see that $\hat{B}$ naturally extends the estimated block probability matrix used in \citep{lei2016goodness} from undirected networks to directed networks.

The estimated expected adjacency matrix \(\hat{\Omega}\) has entries \(\hat{\Omega}(i,j) = \hat{B}(\hat{g}^s(i), \hat{g}^r(j))\) for \(i \neq j\) and \(\hat{\Omega}(i,i) = 0\). 
Then we construct the normalized residual matrix \(\hat{R} \in \mathbb{R}^{n \times n}\) as 
    \begin{equation}\label{eq:residual_matrix}
    \hat{R}(i,j) = 
    \begin{cases} 
        \dfrac{A(i,j) - \hat{\Omega}(i,j)}{\sqrt{(n-1) \hat{\Omega}(i,j) (1 - \hat{\Omega}(i,j))}} & i \neq j, \\
        0 & i = j.
    \end{cases}
    \end{equation}

When $(\hat{B}, \hat{g}^{s}, \hat{g}^{r})$ are accurate enough under the null hypothesis $(K_{s},K_{r})=(K_{s0}, K_{r0})$, $\hat{R}$'s largest singular value should be a good estimation of $\sigma_{1}(R)$. Therefore, we use $\sigma_{1}(\hat{R})$ to define our practical test statistic as follows:
\begin{align}\label{eq:test_stat}  
\hat{T}_n = \sigma_1(\hat{R}) - 2.  
\end{align}   

The normalization in \(\hat{R}\) ensures \(\mathbb{E}[\hat{R}(i,j)] \approx 0\) and \(\operatorname{Var}[\hat{R}(i,j)] \approx (n-1)^{-1}\) under \(H_0\). The shift by 2 anticipates the asymptotic behavior of \(\sigma_1(\hat{R})\)'s upper bound under a correctly specified model, where \(\sigma_1(\hat{R})\)'s upper bound concentrates near 2 as shown by Theorem \ref{thm:null} given later. Under \(H_1\), \(\sigma_1(\hat{R})\) diverges due to unmodeled community structure guaranteed by Theorem \ref{thm:power} given later.  

While \(\mathcal{M}\) can be any method with consistent community recovery for directed networks under \(H_0\), this paper employs a simple spectral clustering approach for its theoretical guarantees and computational efficiency. Algorithm \ref{alg:spectral} details this procedure. 

\begin{algorithm}[H]  
\caption{Spectral clustering for ScBM}\label{alg:spectral}  
\begin{algorithmic}[1]  
\Require Adjacency matrix \(A\), hypothesized community numbers pair \((K_{s0}, K_{r0})\)
\Ensure Estimated labels \(\hat{g}^s, \hat{g}^r\)  
\State Compute $A$'s truncated singular value decomposition: \(A \approx U \Sigma V^\top\) where \(U \in \mathbb{R}^{n \times K_{\mathrm{min}}}\), \(V \in \mathbb{R}^{n \times K_{\mathrm{min}} }\), \(U'U=I_{K_{\mathrm{min}} \times K_{\mathrm{min}} }\), \(V'V=I_{K_{\mathrm{min}} \times K_{\mathrm{min}} }\), \(K_{\mathrm{min}} = \min(K_{s0}, K_{r0})\), and $I_{K_{\mathrm{min}} \times K_{\mathrm{min}} }$ is a $K_{\mathrm{min}}$-by-$K_{\mathrm{min}}$ identity matrix 
\State Apply \(k\)-means with \(K_{s0}\) clusters to rows of \(U\) to obtain \(\hat{g}^s\)
\State Apply \(k\)-means with \(K_{r0}\) clusters to rows of \(V\) to obtain \(\hat{g}^r\) 
\end{algorithmic}  
\end{algorithm}  

This algorithm leverages the fact that, under ScBM, the singular vectors \(U\) and \(V\) encode the sender and receiver community structures, respectively \citep{lei2016goodness}. The \(k\)-means step projects $U$ and $V$ into discrete labels, accommodating asymmetric sender-receiver roles. In fact, Algorithm \ref{alg:spectral} naturally extends the algorithm 1 studied in \citep{lei2015consistency} from SBM to ScBM, and it is also the non-regularized and non-normalized version of the DI-SIM algorithm proposed in \citep{rohe2016co}. Its estimation consistency under ScBM can be found in \citep{rohe2016co,qing2023community,guo2023randomized}.  

Under \(H_0\), consistent estimation of communities ensures \(\hat{R}\) approximates \(R\).  We adopt \citep{lei2016goodness}'s definition of consistency of community recovery.

\begin{defin}[Consistency]
\(\hat{g}^s\) and \(\hat{g}^r\) are consistent if \(\mathbb{P}(\hat{g}^s = g^s) \to 1\) and \(\mathbb{P}(\hat{g}^r = g^r) \to 1\) under \(H_0\).
\end{defin}
Then, we have the following asymptotic result.
\begin{thm}\label{thm:null}
Under $H_0$ and Assumptions \ref{assump:a1}-\ref{assump:a3}, and let $\hat{R}$ be obtained using consistent community estimators $\hat{g}^{s}$ and $\hat{g}^{r}$, then for any $\epsilon>0$, we have
\[
\mathbb{P}(\hat{T}_n < \epsilon) \to 1 \quad \text{as} \quad n \to \infty.
\]
\end{thm}
Theorem \ref{thm:null} guarantees that the upper bound of $\hat{T}_{n}$ converges to 0 in probability under \(H_0\). This provides the theoretical basis for decision rules: small values of $\hat{T}_n$ support $H_0$, while large values signal model inadequacy. The following theorem guarantees that the test statistic diverges for underfitted models where \(K_s > K_{s0}\) or \(K_r > K_{r0}\). 
\begin{thm}\label{thm:power}
Under Assumptions \ref{assump:a1}-\ref{assump:a3}, and if either $K_s > K_{s0}$ or $K_r > K_{r0}$ (or both), then
\[
\hat{T}_n \xrightarrow{P} \infty.
\]
\end{thm}
This divergence property ensures that whenever the hypothesized model lacks sufficient sender or receiver communities, $\hat{T}_n$ will exceed any fixed threshold with probability approaching 1. Combined with Theorem \ref{thm:null}, this guarantees asymptotically separation between correctly specified and underfitted models.
\section{Sequential estimation of community numbers}\label{sec:estimation}
Building on the goodness-of-fit test developed in Section \ref{sec:test}, we now address the core problem of joint community numbers estimation. The sequential testing framework leverages the asymptotic behavior of $\hat{T}_n$ established in Theorems \ref{thm:null} and \ref{thm:power} to systematically identify the true $(K_s, K_r)$ while maintaining computational efficiency. This approach transforms model selection into an ordered exploration of candidate pairs, where the test statistic's dichotomous behavior—convergence to zero under correct specification versus divergence under underfitting—provides reliable stopping criteria. We further propose a ratio-based variant  of this approach to enhance robustness in practical settings.
\subsection{Sequential testing algorithm and its estimation consistency}
The estimation procedure evaluates candidate pairs $(k_s, k_r)$ in \emph{lexicographical order}:
\[
(k_s, k_r) <_{\text{lex}} (k_s', k_r') \iff k_s + k_r < k_s' + k_r' \quad \text{or} \quad k_s + k_r = k_s' + k_r' \ \text{and} \ k_s < k_s'.
\]

This order prioritizes candidate pairs with smaller total community numbers \(k_s + k_r\). Among pairs with the same total, those with a smaller sender community number \(k_s\) are prioritized. Let $\mathcal{P} = \{(k_s, k_r)\}_{m=1}^M$ be the lexicographically ordered sequence of candidate pairs from $(1,1)$ to $(K_{\max}, K_{\max})$, where $M = K_{\max}^2$. The following example details of the lexicographical order of candidate pairs \(\mathcal{P}\) for \(K_{\max} = 8\).
\begin{Ex}\label{example}
Consider searching over candidate sender community numbers \(k_s\) from 1 to 8 and receiver community numbers \(k_r\) from 1 to 8. Table \ref{tab:lex_order} provides the lexicographical order of candidate pairs \(\mathcal{P}\) for \(K_{\max} = 8\).
\begin{table}[htbp]
\centering
\caption{Lexicographical order of candidate pairs $\mathcal{P}$ for $K_{\max}=8$.}
\label{tab:lex_order}
\begin{tabular}{cc|cc|cc|cc|cc|cc|cc|cc}
\toprule
$m$ & $(k_s, k_r)$ & $m$ & $(k_s, k_r)$ & $m$ & $(k_s, k_r)$ & $m$ & $(k_s, k_r)$ & $m$ & $(k_s, k_r)$ & $m$ & $(k_s, k_r)$ & $m$ & $(k_s, k_r)$ & $m$ & $(k_s, k_r)$ \\
\midrule
1 & (1,1) & 9 & (3,2) & 17 & (2,5) & 25 & (4,4) & 33 & (5,4) & 41 & (6,4) & 49 & (8,3) & 57 & (7,6) \\
2 & (1,2) & 10 & (4,1) & 18 & (3,4) & 26 & (5,3) & 34 & (6,3) & 42 & (7,3) & 50 & (4,8) & 58 & (8,5) \\
3 & (2,1) & 11 & (1,5) & 19 & (4,3) & 27 & (6,2) & 35 & (7,2) & 43 & (8,2) & 51 & (5,7) & 59 & (6,8) \\
4 & (1,3) & 12 & (2,4) & 20 & (5,2) & 28 & (7,1) & 36 & (8,1) & 44 & (3,8) & 52 & (6,6) & 60 & (7,7) \\
5 & (2,2) & 13 & (3,3) & 21 & (6,1) & 29 & (1,8) & 37 & (2,8) & 45 & (4,7) & 53 & (7,5) & 61 & (8,6) \\
6 & (3,1) & 14 & (4,2) & 22 & (1,7) & 30 & (2,7) & 38 & (3,7) & 46 & (5,6) & 54 & (8,4) & 62 & (7,8) \\
7 & (1,4) & 15 & (5,1) & 23 & (2,6) & 31 & (3,6) & 39 & (4,6) & 47 & (6,5) & 55 & (5,8) & 63 & (8,7) \\
8 & (2,3) & 16 & (1,6) & 24 & (3,5) & 32 & (4,5) & 40 & (5,5) & 48 & (7,4) & 56 & (6,7) & 64 & (8,8) \\
\bottomrule
\end{tabular}
\end{table}
\end{Ex}

Let $(k_{s}, k_{r})$ be the $m$-th candidate pair in the candidate sequence $\mathcal{P}$. The estimator is 
\[
(\hat{K}_s, \hat{K}_r) = \min\nolimits_{\mathrm{lex}} \left\{ m \in \{1,2,\ldots,M^{2}\}: \hat{T}_n(k_s, k_r) < t_n \right\},
\]
where $t_n = n^{-\varepsilon}$ for $\varepsilon\in(0,0.5)$ is a decaying threshold. Algorithm \ref{alg:DiGoF} below summarizes the details of this sequential testing procedure.
\begin{algorithm}[H]
\caption{DiGoF}\label{alg:DiGoF}
\begin{algorithmic}[1]
\Require Adjacency matrix $A$, significance threshold $t_n$ (default $t_n = n^{-1/5}$), maximum candidate number $K_{\max}$ (default $K_{\max} =  \left\lfloor \sqrt{n / \log n} \right\rfloor$), where $n$ is the number of nodes
\Ensure Estimated community numbers $(\hat{K}_s, \hat{K}_r)$
\State Generate candidate sequence $\mathcal{P} = [(k_s, k_r)]_{m=1}^M$ with $M = K_{\max}^2$ in lexicographical order
\For{$m = 1$ to $M$}
    \State Let $(k_s, k_r) \gets \mathcal{P}(m)$
    \State Compute $\hat{T}_n(k_s, k_r)$ via Equation (\ref{eq:test_stat})
    \If{$\hat{T}_n(k_s, k_r) < t_n$}
        \State \Return $(\hat{K}_s, \hat{K}_r) = (k_s, k_r)$
    \EndIf
\EndFor
\State \Return $(\hat{K}_s, \hat{K}_r) = \mathcal{P}(M)$ \Comment{If no candidate satisfies $\hat{T}_n < t_n$, return the largest candidate}
\end{algorithmic}
\end{algorithm}

$K_\mathrm{max}$ in Algorithm \ref{alg:DiGoF} is set as $ \left\lfloor \sqrt{\frac{n}{\log n}} \right\rfloor$, where this value is suggested by Assumption \ref{assump:a3}. The algorithm starts at \((1,1)\) and sequentially evaluates the test statistic \(\hat{T}_n(k_s, k_r)\) for each candidate pair in the lexicographical order. It stops at the first pair \((k_s, k_r)\) where \(\hat{T}_n(k_s, k_r) < t_n\) and returns \((\hat{K}_s, \hat{K}_r) = (k_s, k_r)\) as the estimated community numbers.
\begin{rem}
(\emph{Rationale for the lexicographical order}). The lexicographical search order \((k_s, k_r) <_{\text{lex}} (k'_s, k'_r)\) is designed to align with the asymptotic behavior of the test statistic \(\hat{T}_n\) shown in Theorems \ref{thm:null}-\ref{thm:power} and ensure computational efficiency:  
\begin{itemize}
  \item It prioritizes candidate pairs with minimal total communities \(k_s + k_r\). Since underfitted models (\(k_s < K_s\) or \(k_r < K_r\)) yield \(\hat{T}_n \stackrel{P}{\to} \infty\) by Theorem \ref{thm:power}, they are quickly rejected. The algorithm only proceeds to higher-complexity candidates when necessary, stopping at the first pair where \(\hat{T}_n < t_n\) (likely near the true \(K_s + K_r\)).  
  \item By testing simpler models first (e.g., (1,1) → (1,2) → (2,1)), the algorithm avoids expensive computations for overfitted candidates (e.g., (5,10)) if earlier pairs fit the data.  
  \item It embodies Occam’s razor by favoring lesser parameterizations (\(k_s \times k_r\) blocks), reducing the risk of overfitting. 
\end{itemize}
\end{rem}
\begin{rem}
It is crucial to emphasize that the sequential search algorithm operates solely based on the established convergence of the upper bound of the test statistic $\hat{T}_n$ to zero under the null hypothesis, without requiring characterization of its lower bound or asymptotic distribution. The algorithm accepts a candidate pair $(k_s, k_r)$ exclusively when $\hat{T}_n$ falls below the decaying threshold $t_n$—an event guaranteed with high probability under the true model $(K_s, K_r)$. While characterizing the precise concentration (including lower bounds) or limiting distribution of $\hat{T}_n$ under $H_0$ remains theoretically valuable, such knowledge is unnecessary for the algorithm's consistency. The divergence of $\hat{T}_n$ to infinity under alternative hypotheses further ensures that underfitted models are rejected with high probability, independently of $\hat{T}_n$'s lower-tail behavior or limit distribution under $H_0$.
\end{rem}

The following theorem guarantees that the sequential procedure achieves joint consistency in recovering the true community numbers under ScBM.
\begin{thm}\label{thm:consistency}
Under Assumptions \ref{assump:a1}-\ref{assump:a3}, let $(\hat{K}_s, \hat{K}_r)$ be obtained from Algorithm \ref{alg:DiGoF} with $t_n = n^{-\varepsilon}$ for any $\varepsilon >0$, we have
\[
\lim_{n \to \infty} \mathbb{P}\left( (\hat{K}_s, \hat{K}_r) = (K_s, K_r) \right) = 1.
\]
\end{thm}
This result establishes that DiGoF correctly identifies both sender and receiver community counts with probability approaching 1 as $n \to \infty$. The lexicographical search ensures the first acceptable pair coincides with $(K_s, K_r)$ asymptotically.
\begin{rem}\label{rem:threshold-choice}
(\emph{Threshold selection}). Though DiGoF enjoys consistency of estimation for any $\varepsilon>0$ in Theorem \ref{thm:consistency}. Our numerical experiments in Section \ref{sec:experiments} find that DiGoF is robust for $\varepsilon < 0.5$, with accuracy declining when $\varepsilon \geq 0.65$. This occurs because excessively rapid decay ($\varepsilon \geq 0.5$) makes $t_n$ vulnerable to random fluctuations in $\hat{T}_n$. Therefore, we recommend $\varepsilon \in (0, 0.5)$ and set default $\varepsilon = 0.2$ for typical $n$.
\end{rem}
\subsection{Ratio-based variant of DiGoF and its estimation consistency}\label{subsec:ratio}
While DiGoF provides consistent estimation under ScBM, its sensitivity to threshold selection may hinder robustness in real-world directed networks with model misspecification. To mitigate this, we propose a ratio-based estimator that identifies the transition point corresponding to the true community structure in the test statistic sequence, leveraging relative changes rather than absolute magnitudes.

For convenience, for each $m \in \{1, 2,\dots,M\}$, here we let $\hat{T}_n(m)$ denote the $\hat{T}_{n}$ computed for the $m$-th candidate pair in $\mathcal{P}$ using Equation (\ref{eq:test_stat}). Define the ratio statistic for the $m$-th candidate pair as
\begin{equation}\label{eq:ratio}
r_{m} = \left| \frac{\hat{T}_n(m-1)}{\hat{T}_n(m)} \right|, \quad m=2,3,\ldots, M,
\end{equation}
where the absolute value addresses potential negative value of $\hat{T}_{n}$ for the true model $(K_s, K_r)$.

Theorems \ref{thm:null} and \ref{thm:power} guarantee that under the true model \((K_s, K_r)\) (with lexicographic index \(m_*\) in \(\mathcal{P}\)), the upper bound of \(\hat{T}_n(m_*)\) converges to zero with high probability, while \(\hat{T}_n(m_*-1)\) diverges to infinity under underfitting when \(m_* > 1\). Consequently, we shall expect that \(r_{m_*} = |\hat{T}_n(m_*-1)/\hat{T}_n(m_*)|\) is the first significant peak (change point) in the ratio sequence \(\{r_m\}^{M}_{m=2}\). We refer to this method as Ratio-DiGoF (RDiGoF for short), which identifies the first peak in the sequence \(\{(m, r_m)\}_{m=2}^M\). The complete RDiGoF algorithm is summarized in Algorithm \ref{alg:RDIGoF}.

\begin{algorithm}[H]
\caption{RDiGoF}\label{alg:RDIGoF}
\begin{algorithmic}[1]
\Require  Adjacency matrix $A$, threshold $\tau > 0$ (default $\tau = 10$), maximum candidate number $K_{\max}$ (default $K_{\max} =  \left\lfloor \sqrt{n / \log n} \right\rfloor$), where $n$ is the number of nodes
\Ensure Estimated community numbers $(\hat{K}_s, \hat{K}_r)$
\State Generate candidate sequence $\mathcal{P} = [(k_s, k_r)]_{m=1}^M$ with $M = K_{\max}^2$ in lexicographical order
\State Compute $\hat{T}_n(1)$ for candidate $(1,1)$ via Equation (\ref{eq:test_stat})
\If{$\hat{T}_n(1) < n^{-1/5}$}
    \State \Return $(\hat{K}_{s},\hat{K}_{r})=(1,1)$
\EndIf
\For{$m = 2$ to $M$}
    \State Compute ratio statistic $r_m$ via Equation (\ref{eq:ratio})
    \If{$r_m > \tau$}
        \State \Return $(\hat{K}_s, \hat{K}_r) = \mathcal{P}(m)$
    \EndIf
\EndFor
\State \Return $(\hat{K}_s, \hat{K}_r) = \mathcal{P}(M)$ \Comment{If no candidate satisfies $r_m > \tau$, return the largest candidate}
\end{algorithmic}
\end{algorithm}

The following theoretical results establish the estimation consistency of RDiGoF, analogous to Theorems \ref{thm:power} and \ref{thm:consistency}.

\begin{thm}\label{thm:ratio}
Under the assumptions of Theorem \ref{thm:consistency}, and assuming additionally that $\delta_n \geq \delta_{\min} > 0$ for all $n$ and that $K_{\max} = \max(K_s, K_r)$ is fixed, then:
\begin{enumerate}
    \item For underfitted models ($m < m_*$): $r_m = O_P(1)$.
    \item For the true model ($m = m_*$): $r_{m_*} \stackrel{P}{\to} \infty$.
\end{enumerate}
\end{thm}

\begin{thm}\label{thm:RDiGoF-consistency}
Under the assumptions of Theorem \ref{thm:ratio}, let $(\hat{K}_s, \hat{K}_r)$ be the output of Algorithm \ref{alg:RDIGoF}. Then there exists a constant $C > 0$ such that if $\tau > C$,
\[
\lim_{n \to \infty} \mathbb{P}\left( (\hat{K}_s, \hat{K}_r) = (K_s, K_r) \right) = 1.
\]
\end{thm}
\section{Numerical experiments}\label{sec:experiments}
In this section, we conduct comprehensive numerical studies to evaluate four aspects of our method: (1) convergence behavior of $\hat{T}_n$ under $H_0$ to verify Theorem \ref{thm:null}, (2) behavior characteristics of DiGoF under both correctly specified and underfitted models to verify Theorem \ref{thm:power}, (3) estimation accuracy of DiGoF and RDiGoF in estimating the numbers of sender and receiver communities across various ScBM parameter settings, and (4) DiGoF's (and RDiGoF's) sensitivity to the threshold selection $t_n$ (and $\tau$). Unless specified, we set $\varepsilon=1/5$ in DiGoF and $\tau=10$ in RDiGoF. Directed networks are generated under ScBM as follows:
\begin{enumerate}
    \item[(I)] Assign each node to each sender (or receiver) community with equal probability.
    \item[(II)] Unless specified, construct the block probability matrix $B \in [0,1]^{K_s \times K_r}$ as below
 \[
B(k,l) = \rho \cdot \begin{cases} 
\alpha + \beta & \text{if } k = l \text{ and } k \leq \min(K_s, K_r) \\
\beta & \text{otherwise}
\end{cases}
\]
where $\alpha = 0.7$, $\beta = 0.2$, and $\rho \in (0,1)$ is a sparsity parameter that controls the sparsity of a directed network.
    
    \item[(III)] Generate $A$ via Definition \ref{def:ScBM} with independent entries
    \[
    A(i,j) \sim \text{Bernoulli}(B(g^s(i), g^r(j))) \quad \text{for} \quad i \neq j, \quad A(i,i) = 0.
    \]
\end{enumerate}

\subsection{Convergence of $\hat{T}_n$ under null hypothesis}\label{sec:nullconv}
We verify Theorem \ref{thm:null} by simulating $\hat{T}_n$ under $H_0$: $(K_{s0}, K_{r0}) = (K_s, K_r)$. We set network sizes $n\in\{200, 400,\ldots,2000\}$, $(K_s, K_r) \in \{(1,2), (2,2), (2,3), (2,4), (3,4), (3,5)\}$, and sparsity levels $\rho \in \{0.05, 0.1\}$. For each parameter configuration, we generate 200 independent directed networks and compute the mean $\hat{T}_n$ value. As shown in Figure \ref{fig:ex1}, $\hat{T}_n$ converges to zero as $n$ increases across all configurations, validating Theorem \ref{thm:null}. Meanwhile, convergence rates exhibit significant dependence on network sparsity ($\rho$). For $\rho = 0.05$, $\hat{T}_n$ decays more slowly than under $\rho = 0.1$, reflecting greater variance in the residual matrix under sparser regimes.
\begin{figure}[htp!]
\centering
\resizebox{\columnwidth}{!}{
\subfigure{\includegraphics[width=0.45\textwidth]{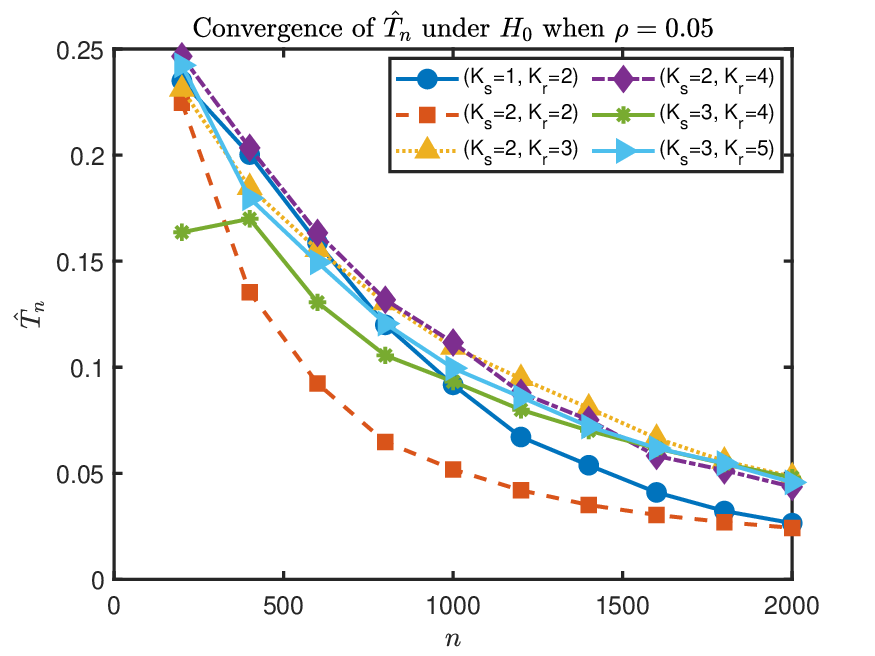}}
\subfigure{\includegraphics[width=0.45\textwidth]{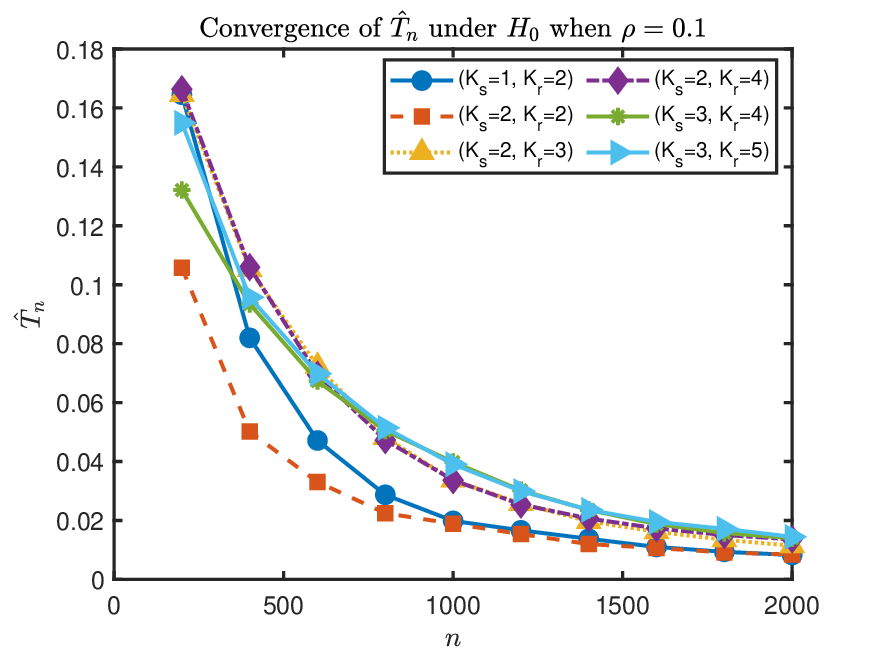}}
}
\caption{Behavior of $\hat{T}_{n}$ under $H_{0}$ as network size $n$ increases in different configurations of sender/receiver communities and network sparsity.}
\label{fig:ex1} 
\end{figure}
\subsection{DiGoF's behavior under true and underfitted models}\label{sec:size_power}
To validate the theoretical properties in Theorems \ref{thm:null} and \ref{thm:power}, we quantify DiGoF's acceptance rate under true models and rejection rate under underfitted models. The acceptance rate measures the probability that DiGoF accepts the true community pair $(K_{s0}, K_{r0}) = (K_s, K_r)$ when $\hat{T}_n < t_n$, while the rejection rate captures the probability of $\hat{T}_n \geq t_n$ when either $K_{s0} < K_s$ or $K_{r0} < K_r$. We fix $n = 1000$ and $\rho = 0.1$, with true community structures $(K_s, K_r) \in \{(2,3), (3,3), (3,4)\}$. For each true structure, we evaluate all candidate pairs $(K_{s0}, K_{r0})$ satisfying $K_{s0} \leq K_s$, $K_{r0} \leq K_r$, and $(K_{s0}, K_{r0}) \neq (K_s, K_r)$. Each configuration undergoes 200 replications. Table \ref{tab:size_power} summarizes the results, showing 100\% acceptance rates for correctly specified models and 100\% rejection rates for underfitted models, confirming theoretical predictions across sender-only, receiver-only, and joint underfitting scenarios.  This sharp dichotomy between perfect rejection of underfitted models and full acceptance of the true model underscores the statistical discriminative power of the proposed goodness-of-fit test.
\begin{table}[htbp]
\centering
\caption{Rejection rate for different true and hypothesized community pairs.}\label{tab:size_power}
\begin{tabular}{@{}cccc@{}}
\toprule
True \((K_s, K_r)\) & Hypothesized \((K_{s0}, K_{r0})\) & Rejection rate & Underfitting type \\
\midrule
\multirow{6}{*}{ (2,3) }
& (1,1) & 1.00 & Sender \& receiver \\
& (1,2) & 1.00 & Sender \& receiver \\
& (1,3) & 1.00 & Sender only \\
& (2,1) & 1.00 & Receiver only \\
& (2,2) & 1.00 & Receiver only \\
& (2,3) & 0    & -- \\
\midrule
\multirow{9}{*}{ (3,3) }
& (1,1) & 1.00 & Sender \& receiver \\
& (1,2) & 1.00 & Sender \& receiver \\
& (1,3) & 1.00 & Sender only \\
& (2,1) & 1.00 & Sender \& receiver \\
& (2,2) & 1.00 & Sender \& receiver \\
& (2,3) & 1.00 & Sender only \\
& (3,1) & 1.00 & Receiver only \\
& (3,2) & 1.00 & Receiver only \\
& (3,3) & 0    & -- \\
\midrule
\multirow{12}{*}{ (3,4) }
& (1,1) & 1.00 & Sender \& receiver \\
& (1,2) & 1.00 & Sender \& receiver \\
& (1,3) & 1.00 & Sender \& receiver \\
& (1,4) & 1.00 & Sender only \\
& (2,1) & 1.00 & Sender \& receiver \\
& (2,2) & 1.00 & Sender \& receiver \\
& (2,3) & 1.00 & Sender \& receiver \\
& (2,4) & 1.00 & Sender only \\
& (3,1) & 1.00 & Receiver only \\
& (3,2) & 1.00 & Receiver only \\
& (3,3) & 1.00 & Receiver only \\
& (3,4) & 0    & -- \\
\bottomrule
\end{tabular}
\end{table}
\subsection{Accuracy of DiGoF and RDiGoF}
We evaluate the accuracy of DiGoF and RDiGoF in estimating sender and receiver community numbers under varied parameter settings. We set true community number pairs $(K_s, K_r) \in \{(1,1), (1,2), (2,3), (3,3), (3,4), (5,4)\}$ to maintain $|K_s-K_r|\leq1$ to make $B$ satisfy Assumption \ref{assump:a3}. Network sizes range over $n \in \{200, 400, 600,800,1000\}$ with sparsity levels $\rho \in \{0.1, 0.2, 0.3, 0.4\}$. Table \ref{tab:accuracy} reports the empirical accuracy $\mathbb{P}((\hat{K}_s, \hat{K}_r) = (K_s, K_r))$ computed over 100 independent replications. The results shown in Table \ref{tab:accuracy} reveal critical performance trends under varying ScBM configurations. For simple structures like \((1,1)\) and \((1,2)\), DiGoF achieves perfect accuracy (1.00) across all \(n\) and \(\rho\), demonstrating robustness to sparsity when community counts are low. However, for complex structures (e.g., \((2,3)\), \((3,4)\), \((5,4)\)), accuracy exhibits a phase transition: at \(n=200\) and low sparsity (\(\rho=0.1\)), accuracy drops sharply (e.g., \(0.00\) for \((2,3)\) and \((5,4)\)), indicating insufficient signal for detection. As network size \(n\) (or sparsity parameter $\rho$) increases, accuracy improves monotonically. This confirms that larger networks and denser edges enhance community recovery, with DiGoF achieving asymptotic consistency shown in Theorem \ref{thm:consistency}. Furthermore, based on Table \ref{tab:accuracy}, RDiGoF demonstrates comparable accuracy to DiGoF in most scenarios, particularly as network size or density increases. While it exhibits slight underperformance in very sparse/small settings (e.g., (1,2) at $n=200, \rho=0.1$), it matches or marginally exceeds DiGoF in more complex or denser cases (e.g., (5,4) at $n=600, \rho=0.2$). Overall, both methods converge rapidly to near-perfect accuracy under sufficient signal.
\begin{table}[htbp]
\centering
\caption{Accuracy of DiGoF and RDiGoF under varying sparsity levels $\rho$ for different network sizes $n$ and community number pairs $(K_s, K_r)$.}\label{tab:accuracy}
\small
\begin{tabular}{@{}cc cccc cccc@{}}
\toprule
\multirow{2}{*}{$(K_s, K_r)$} & \multirow{2}{*}{$n$} & \multicolumn{4}{c}{DiGoF} & \multicolumn{4}{c}{RDiGoF} \\
\cmidrule(lr){3-6} \cmidrule(lr){7-10}
 & & $\rho=0.1$ & $\rho=0.2$ & $\rho=0.3$ & $\rho=0.4$ & $\rho=0.1$ & $\rho=0.2$ & $\rho=0.3$ & $\rho=0.4$ \\
\midrule
\multirow{5}{*}{ (1,1) } 
& 200  & 1.00 & 1.00 & 1.00 & 1.00 & 1.00 & 1.00 & 1.00 & 1.00 \\
& 400  & 1.00 & 1.00 & 1.00 & 1.00 & 1.00 & 1.00 & 1.00 & 1.00 \\
& 600  & 1.00 & 1.00 & 1.00 & 1.00 & 1.00 & 1.00 & 1.00 & 1.00 \\
& 800  & 1.00 & 1.00 & 1.00 & 1.00 & 1.00 & 1.00 & 1.00 & 1.00 \\
& 1000 & 1.00 & 1.00 & 1.00 & 1.00 & 1.00 & 1.00 & 1.00 & 1.00 \\
\addlinespace
\hline
\multirow{5}{*}{ (1,2) } 
& 200  & 1.00 & 1.00 & 1.00 & 1.00 & 0.65 & 1.00 & 1.00 & 1.00 \\
& 400  & 1.00 & 1.00 & 1.00 & 1.00 & 1.00 & 1.00 & 1.00 & 1.00 \\
& 600  & 1.00 & 1.00 & 1.00 & 1.00 & 1.00 & 1.00 & 1.00 & 1.00 \\
& 800  & 1.00 & 1.00 & 1.00 & 1.00 & 1.00 & 1.00 & 1.00 & 1.00 \\
& 1000 & 1.00 & 1.00 & 1.00 & 1.00 & 1.00 & 1.00 & 1.00 & 1.00 \\
\addlinespace
\hline
\multirow{5}{*}{ (2,3) } 
& 200  & 0.10 & 0.43 & 0.99 & 1.00 & 0.06 & 0.83 & 1.00 & 1.00 \\
& 400  & 0.75 & 1.00 & 1.00 & 1.00 & 0.94 & 1.00 & 1.00 & 1.00 \\
& 600  & 1.00 & 1.00 & 1.00 & 1.00 & 1.00 & 1.00 & 1.00 & 1.00 \\
& 800  & 1.00 & 1.00 & 1.00 & 1.00 & 1.00 & 1.00 & 1.00 & 1.00 \\
& 1000 & 1.00 & 1.00 & 1.00 & 1.00 & 1.00 & 1.00 & 1.00 & 1.00 \\
\addlinespace
\hline   
\multirow{5}{*}{ (3,3) } 
& 200  & 0.00 & 0.99 & 1.00 & 1.00 & 0.16 & 1.00 & 1.00 & 1.00 \\
& 400  & 1.00 & 1.00 & 1.00 & 1.00 & 0.98 & 1.00 & 1.00 & 1.00 \\
& 600  & 1.00 & 1.00 & 1.00 & 1.00 & 1.00 & 1.00 & 1.00 & 1.00 \\
& 800  & 1.00 & 1.00 & 1.00 & 1.00 & 1.00 & 1.00 & 1.00 & 1.00 \\
& 1000 & 1.00 & 1.00 & 1.00 & 1.00 & 1.00 & 1.00 & 1.00 & 1.00 \\
\addlinespace
\hline   
\multirow{5}{*}{ (3,4) } 
& 200  & 0.00 & 0.00 & 0.03 & 0.64 & 0.06 & 0.11 & 0.45 & 0.83 \\
& 400  & 0.00 & 0.67 & 1.00 & 1.00 & 0.12 & 0.82 & 1.00 & 1.00 \\
& 600  & 0.01 & 1.00 & 1.00 & 1.00 & 0.01 & 1.00 & 1.00 & 1.00 \\
& 800  & 0.87 & 1.00 & 1.00 & 1.00 & 0.93 & 1.00 & 1.00 & 1.00 \\
& 1000 & 1.00 & 1.00 & 1.00 & 1.00 & 1.00 & 1.00 & 1.00 & 1.00 \\
\addlinespace
\hline
\multirow{5}{*}{ (5,4) } 
& 200  & 0.00 & 0.00 & 0.00 & 0.00 & 0.00 & 0.00 & 0.03 & 0.19 \\
& 400  & 0.00 & 0.00 & 0.45 & 1.00 & 0.00 & 0.04 & 0.62 & 0.99 \\
& 600  & 0.00 & 0.47 & 1.00 & 1.00 & 0.00 & 0.61 & 0.99 & 1.00 \\
& 800  & 0.00 & 1.00 & 1.00 & 1.00 & 0.02 & 0.99 & 1.00 & 1.00 \\
& 1000 & 0.01 & 1.00 & 1.00 & 1.00 & 0.15 & 1.00 & 1.00 & 1.00 \\
\addlinespace
\bottomrule
\end{tabular}
\end{table}
\subsection{Robustness to threshold choice}
We assess DiGoF's sensitivity to the threshold parameter $t_n = n^{-\varepsilon}$ using fixed network sizes $n\in\{600,1200\}$, true community structure $(K_s, K_r) = (2, 4)$, and sparsity $\rho = 0.5$. The block probability matrix is explicitly defined as:
\[
B = \rho \cdot \begin{bmatrix} 
0.1 & 0.9 & 0.4 & 0.1 \\
0.7 & 0.1 & 0.6 & 0.3 
\end{bmatrix}.
\]
The threshold parameter $\varepsilon$ varies from $0.05$ to $1.00$ in $0.05$ increments, yielding 20 distinct thresholds $t_n = n^{-\varepsilon}$. Accuracy is evaluated over 100 independent replications per $\varepsilon$ value. The numerical results displayed in Figure \ref{fig:ex4} demonstrate that DiGoF maintains high accuracy ($\geq$0.95) across a broad threshold range $\varepsilon \in (0.05, 0.55]$, with accuracy decline observed only for $\varepsilon \geq0.65$. This decline aligns with the convergence behavior of the upper bound of $\hat{T}_n$ shown in Figure \ref{fig:ex1}. In detail, for large $\varepsilon$, the threshold $t_n$ decays excessively fast, causing the algorithm to occasionally fail due to the random fluctuations in $\hat{T}_n$. Though our theoretical analysis in Theorem \ref{thm:consistency} shows that DiGoF can consistently determine the numbers of sender and receiver communities for any $\varepsilon>0$, random fluctuations in $\hat{T}_n$ shown in Figure \ref{fig:ex1} may cause DiGoF to accept underfitted models for large $\varepsilon$. Thus, DiGoF is robust for moderate \(\varepsilon\) (\(<0.5\)) but exhibits sensitivity at larger values. We recommend \( \varepsilon \in (0, 0.5) \) in practice and set the default value of $\varepsilon$ as \(0.2 \) for typical \( n \) in DiGoF.
\begin{figure}[htp!]
\centering
\resizebox{\columnwidth}{!}{
\subfigure{\includegraphics[width=0.5\textwidth]{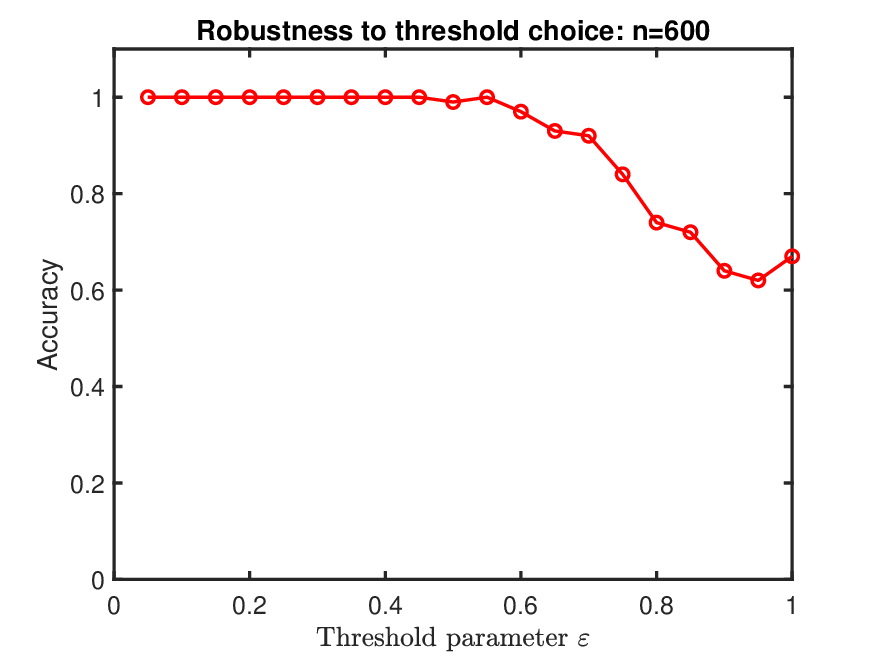}}
\subfigure{\includegraphics[width=0.5\textwidth]{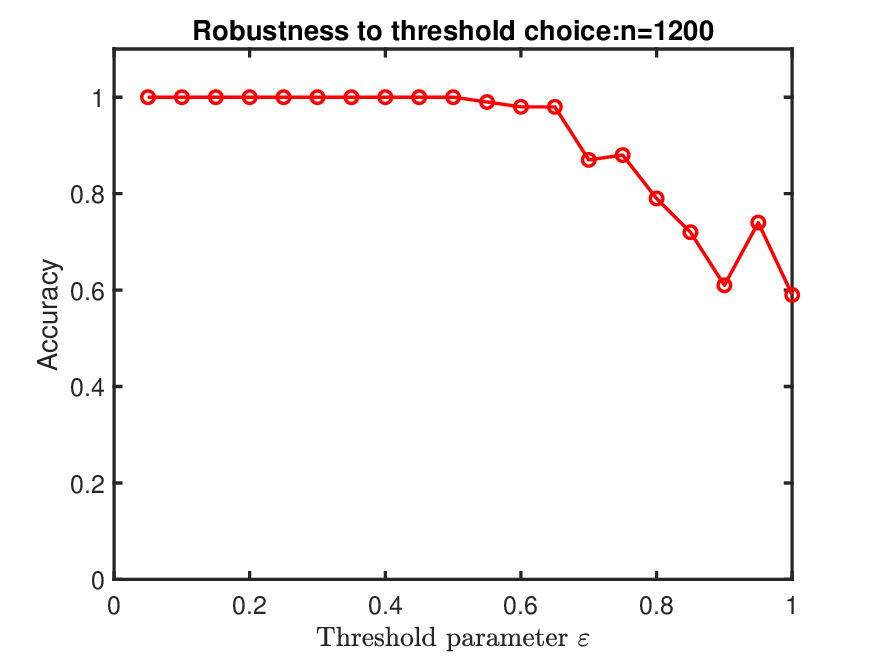}}
}
\caption{Accuracy of DiGoF under varying thresholds.}
\label{fig:ex4} 
\end{figure}

Under the same configuration as in the DiGoF sensitivity analysis, we evaluate the robustness of RDiGoF with respect to the threshold parameter \(\tau\), varying over the set \(\{2, 4, \ldots, 20\}\). The results, displayed in Figure \ref{fig:ex4RDiGoF}, indicate that RDiGoF achieves nearly perfect accuracy in recovering the true community number pair \((K_s, K_r)\) for all \(\tau \geq 4\). This consistency aligns with the theoretical guarantees established in Theorem \ref{thm:RDiGoF-consistency}, which requires \(\tau\) to exceed a certain positive constant \(C\). In contrast, when \(\tau = 2\), the method fails to correctly identify the true community structure. This is expected, as a threshold that is too small may cause the algorithm to mistakenly identify noise fluctuations as significant peaks in the ratio statistic sequence \(\{r_m\}_{m=2}^M\), rather than the true transition point corresponding to the true model. Thus, the poor performance at \(\tau = 2\) underscores the necessity of choosing a sufficiently large \(\tau\) to ensure statistical reliability, as required by the theoretical conditions. These findings confirm that RDiGoF is highly robust to the choice of \(\tau\) within a reasonable range, further supporting its practical utility for asymmetric community numbers estimation in directed networks.
\begin{figure}[htp!]
\centering
\resizebox{\columnwidth}{!}{
\subfigure{\includegraphics[width=0.5\textwidth]{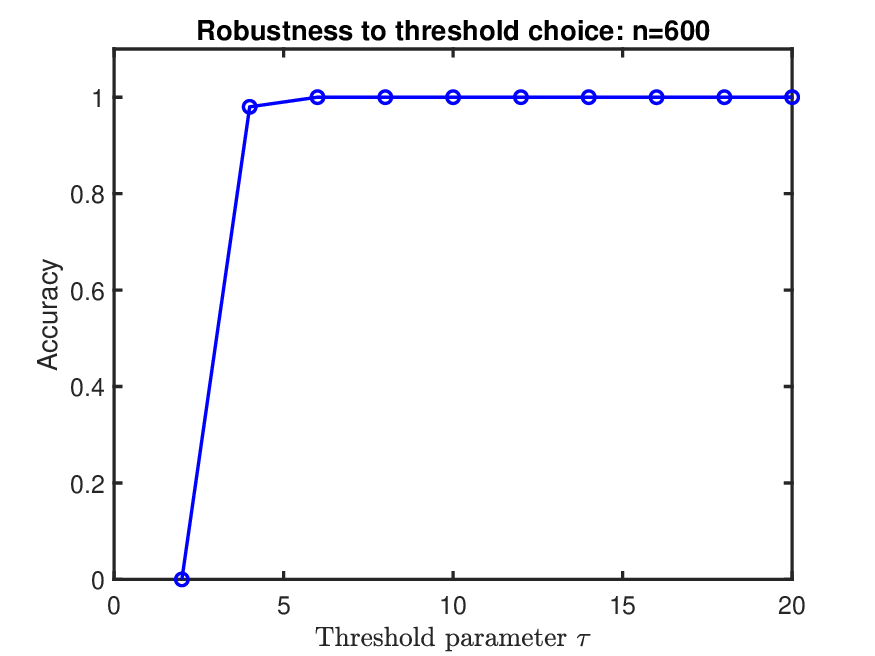}}
\subfigure{\includegraphics[width=0.5\textwidth]{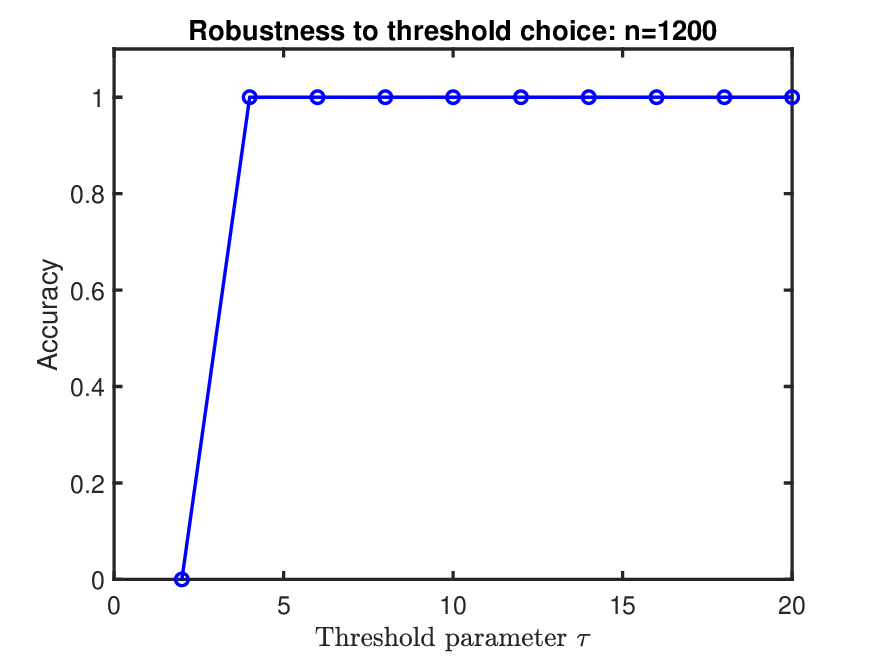}}
}
\caption{Accuracy of RDiGoF under varying thresholds.}
\label{fig:ex4RDiGoF} 
\end{figure}

\begin{rem}
To investigate the behaviors of $|\hat{T}_{n}(m)|$ and $r_{m}$, we present Figure \ref{fig:ex4ratio} which displays the sequence of $|\hat{T}_n(m)|$ values (top panels) and corresponding ratio statistics $r_m$ (bottom panels) for candidate community number pairs under the same simulation settings as Figure \ref{fig:ex4}. Results are shown for $n=600$ (left) and $n=1200$ (right), with true community structure $(K_s,K_r)=(2,4)$. The top panels demonstrate the characteristic transition in $|\hat{T}_n(m)|$: underfitted models ($m < m_*$) exhibit large positive values that diverge as predicted by Theorem \ref{thm:power}, while at the true model $(m = m_* = 12)$, $|\hat{T}_n(m)|$ sharply drops near zero. This transition creates a obvious peak in $\{r_m\}^{M}_{m=2}$ at $m_*$ (bottom panels), where $r_{m_*} = |\hat{T}_n(m_*-1)/\hat{T}_n(m_*)|$ becomes substantially larger than 1 ($r_{m_{*}}>400$ when $n=600$ and $r_{m_{*}}>1500$ when $n=1200$). We see that the first significant peak ($r_m > \tau = 10$) in RDiGoF consistently identifies the true community structure across both network sizes.
\begin{figure}[htp!]
\centering
\resizebox{\columnwidth}{!}{
\subfigure{\includegraphics[width=0.5\textwidth]{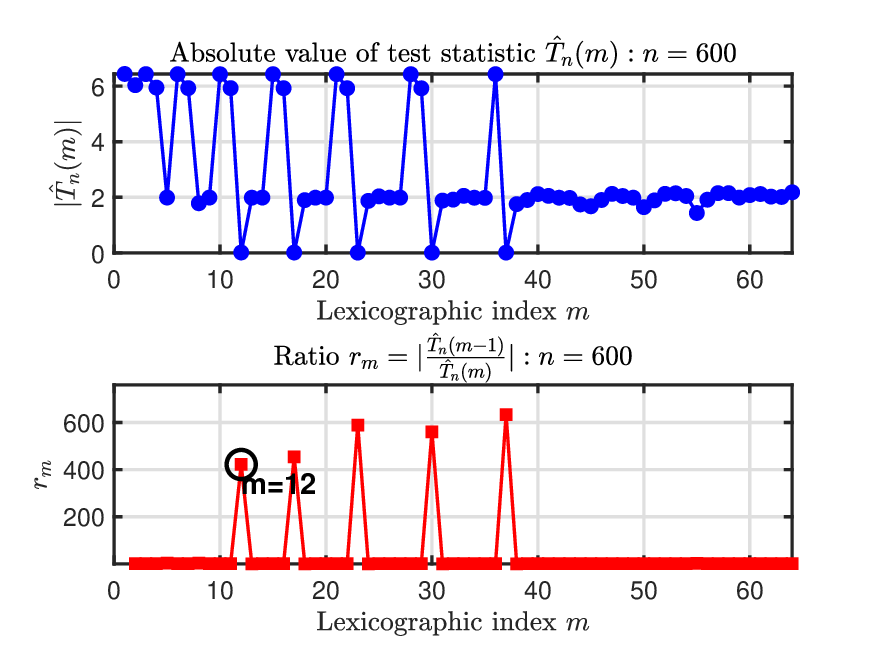}}
\subfigure{\includegraphics[width=0.5\textwidth]{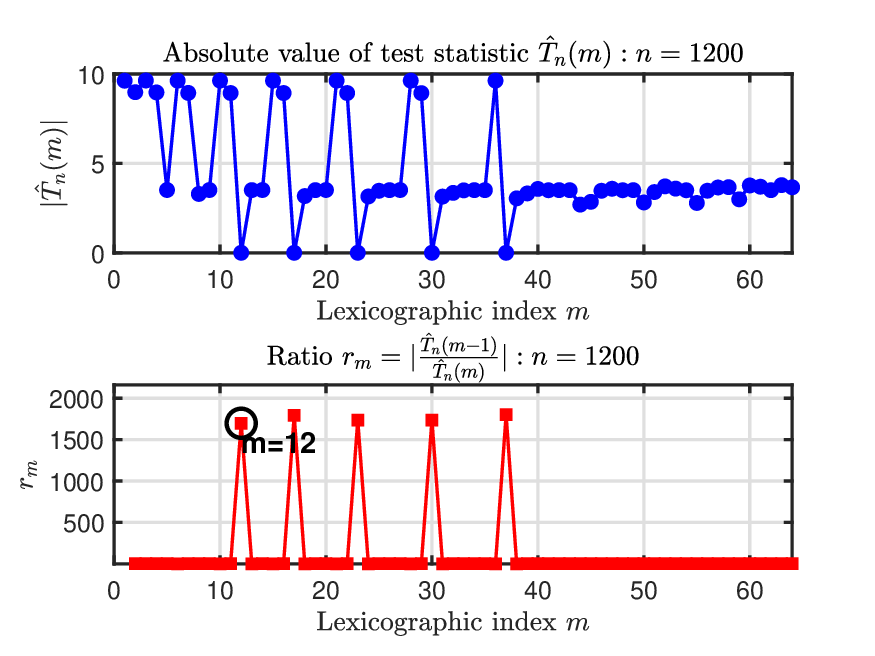}}
}
\caption{Absolute value of test statistic sequence $\hat{T}_n(m)$ (top) and ratio statistic $r_m$ (bottom) for lexicographically ordered candidate pairs $1 \leq m \leq 64$ (i.e., $K_{\mathrm{max}}=8$) at $n=600$ (left) and $n=1200$ (right). Black circles indicate the position of the first peak, which is precisely at $m_{*} = 12$ corresponding to $(K_s,K_r) = (2,4)$ by Table \ref{tab:lex_order}.}
\label{fig:ex4ratio} 
\end{figure}
\end{rem}
\subsection{Real data examples}
In this subsection, we evaluate our method on six real-world directed networks from diverse domains. Basic network statistics are summarized in Table \ref{realdata}, where isolated nodes have been removed and only the largest connected components are retained. All datasets are available at \url{http://konect.cc/networks/}.
\begin{table}[h!]
\small
	\centering
	\caption{Basic information of real-world directed networks.}
	\label{realdata}
\resizebox{\columnwidth}{!}{
	\begin{tabular}{cccccccccccc}
\hline\hline
Dataset&Source&Node meaning&Directed edge meaning&$n$&Number of edges&True $(K_{s}, K_{r})$\\
\hline
Hens&\citep{guhl1953social}&Hen&Dominance&31&465&Unknown\\
Physicians&\citep{coleman1957diffusion}&Physician&Trust&95&458&Unknown\\
FilmTrust&\citep{guo2016novel}&User&Trust&347&1254&Unknown\\
OpenFlights&\citep{opsahl2010node}&Airport&Flight&2868&30404&Unknown\\
Wikipedia links (rue)&\citep{kunegis2013konect}&Wikipedia article in the Rusyn language&Hyperlink&5681&167900&Unknown\\
Wikipedia links (xal)&\citep{kunegis2013konect}&Wikipedia article in the Kalmyk language&Hyperlink&1519&229433&Unknown\\
\hline\hline
\end{tabular}
}
\end{table}
\begin{figure}[htp!]
\centering
\resizebox{\columnwidth}{!}{
\subfigure{\includegraphics[width=0.5\textwidth]{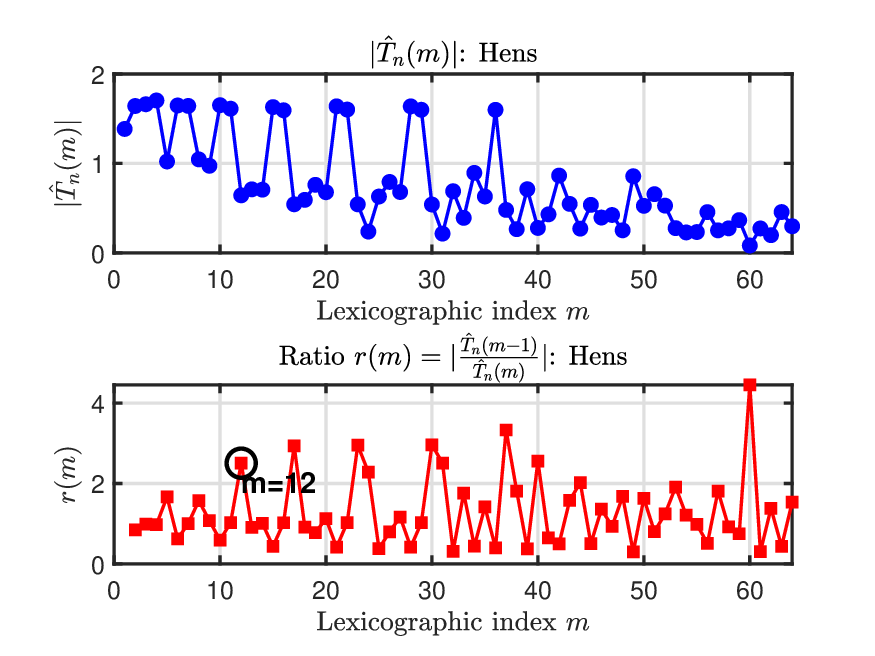}}
\subfigure{\includegraphics[width=0.5\textwidth]{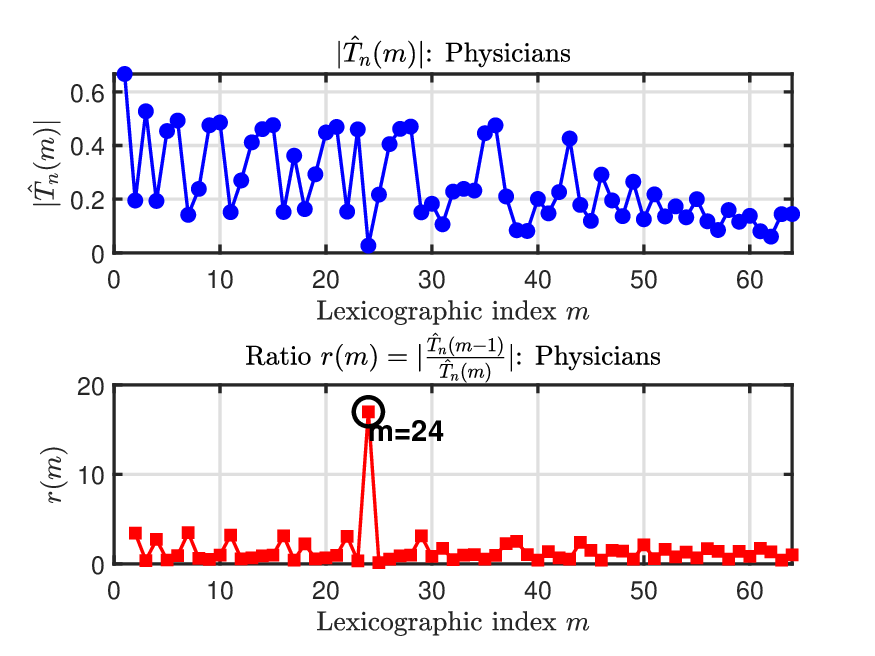}}
\subfigure{\includegraphics[width=0.5\textwidth]{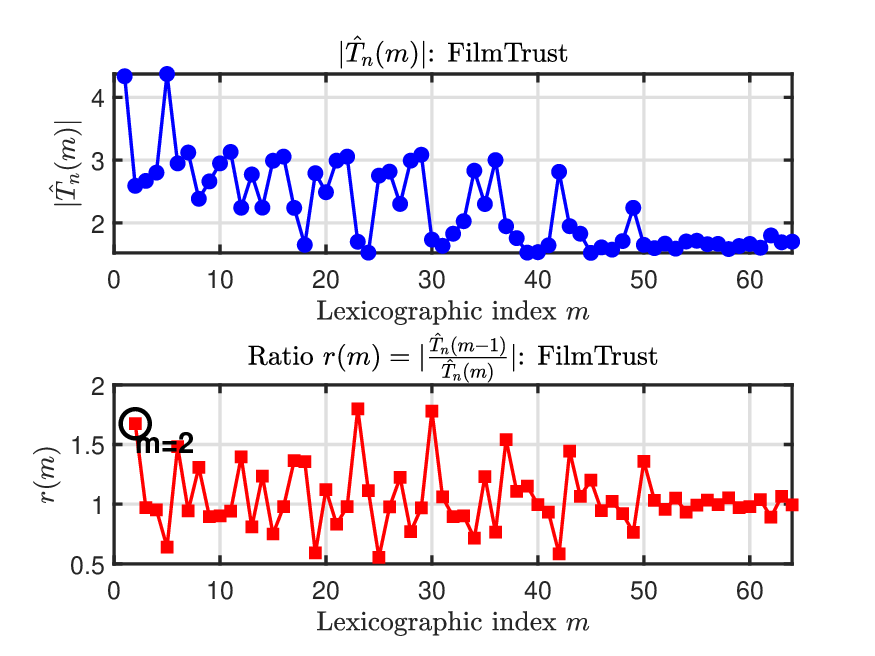}}
}
\resizebox{\columnwidth}{!}{
\subfigure{\includegraphics[width=0.5\textwidth]{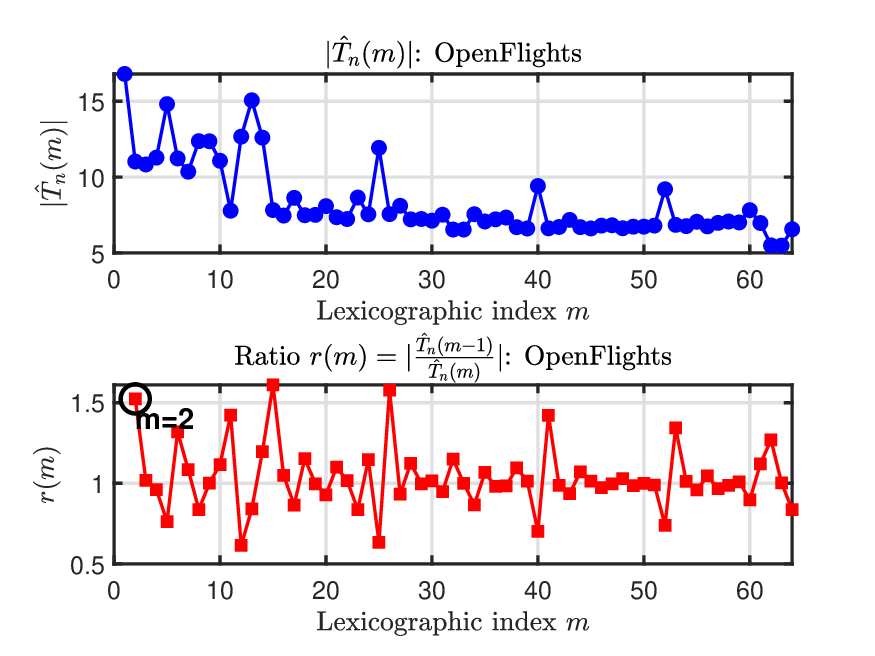}}
\subfigure{\includegraphics[width=0.5\textwidth]{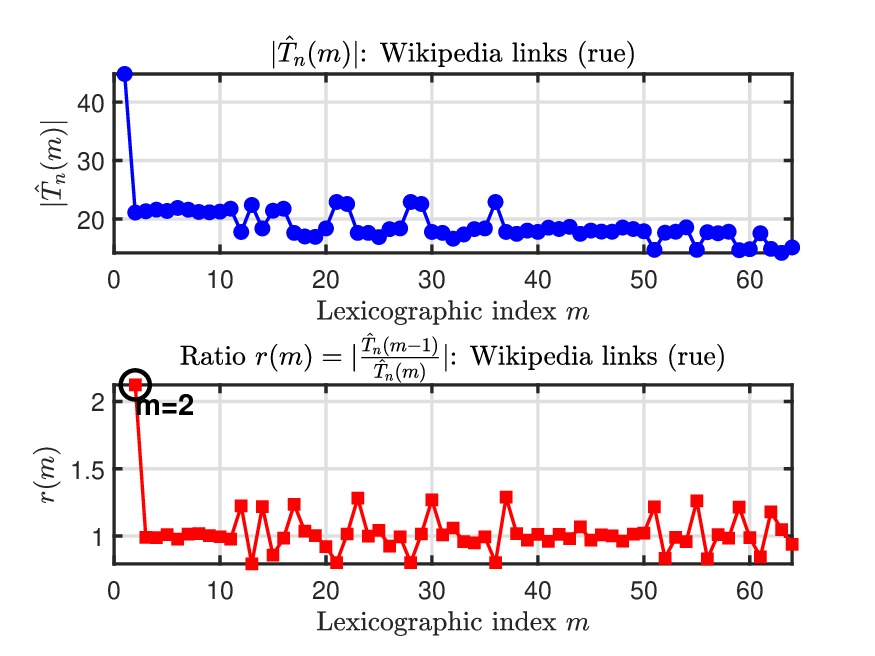}}
\subfigure{\includegraphics[width=0.5\textwidth]{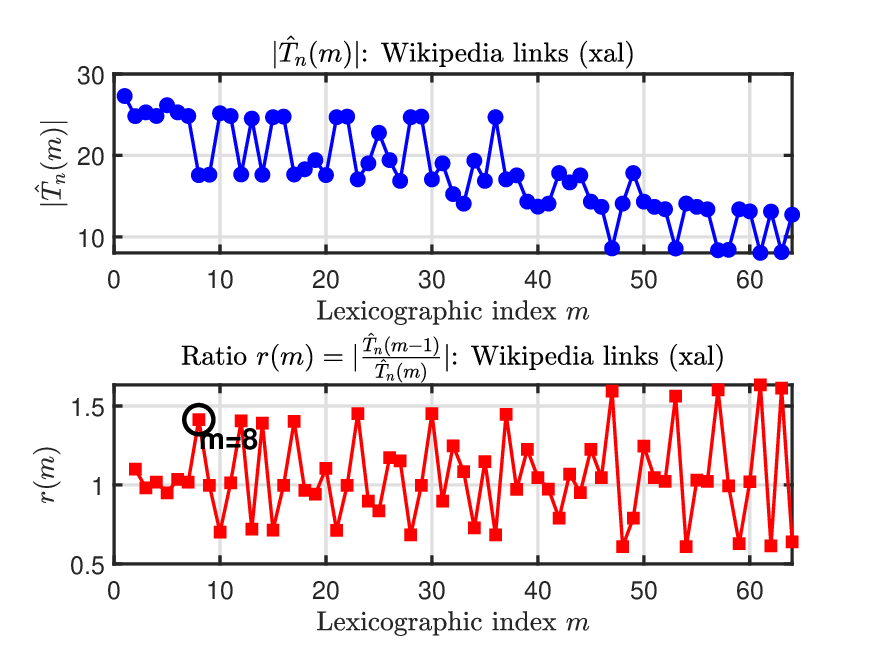}}
}
\caption{$|\hat{T}_n(m)|$ and $r_m$ for lexicographically ordered candidate pairs $1 \leq m \leq 64$ (i.e., $K_{\mathrm{max}}=8$) for real-world directed networks used in this paper. Black circles indicate the first significant peak in $r_m$ identified by RDiGoF.}
\label{fig:realdataratio} 
\end{figure}

Figure \ref{fig:realdataratio} displays $|\hat{T}_n(m)|$ and $r_m$ for all real-world directed networks when $K_{\mathrm{max}}=8$. Unlike simulated directed networks where $r_{m_*}$ exceeds 400 shown in Figure \ref{fig:ex4ratio}, real-world directed networks exhibit substantially smaller ratios: 
\begin{itemize}
   \item $r_{m}<2.5$ for FilmTrust, OpenFlights, and both Wikipedia networks.
   \item $r_m < 5$ for Hens.
    \item $r_m < 15$ for Physicians.
\end{itemize}

Consequently, DiGoF fails to select meaningful community numbers even with $\varepsilon \to 0$, as $|\hat{T}_n(m)|$ never suddenly approaches 0 as shown in the top panels of Figure \ref{fig:realdataratio}. This phenomenon aligns with observations in real-world undirected networks \citep{lei2016goodness}. Through adjusting the threshold parameter $\tau$, RDiGoF overcomes this limitation by identifying the first significant peak in $\{r_m\}^{K^{2}_{\mathrm{max}}}_{m=2}$. By combining the circled peak points in Figure \ref{fig:realdataratio} with the lexicographical order of candidate pairs shown in Table \ref{tab:lex_order}, we can obtain the estimated community numbers returned by RDiGoF. The results are reported in Table \ref{realdataKsKr}. These results demonstrate RDiGoF's superior practicality for real-world directed networks, where DiGoF's absolute threshold criterion is inadequate.
\begin{table}[h!]
\small
	\centering
	\caption{Estimated community numbers and the corresponding lexicographic index returned by RDiGoF for real-world directed networks.}
	\label{realdataKsKr}
\begin{tabular}{c|c|c}
\hline\hline
Directed network & $(\hat{K}_s, \hat{K}_r)$ &Lexicographic index $m_*$ \\ \hline
Hens & (2,4) & 12 \\
Physicians & (3,5) & 24 \\
FilmTrust & (1,2) & 2 \\
OpenFlights & (1,2) & 2 \\
Wiki (rue) & (1,2) & 2 \\
Wiki (xal) & (2,3) & 8 \\
\hline\hline
\end{tabular}
\end{table}
\section{Conclusion}\label{sec:conclusion}
This paper presents a solution to the critical challenge of jointly estimating sender and receiver community numbers in directed networks under the stochastic co-block model (ScBM). Based on the tail bounds of the largest singular value of a normalized residual matrix, we introduce a novel goodness-of-fit test  \(\hat{T}_n = \sigma_1(\hat{R}) - 2\), whose asymptotic behavior differs under null and alternative hypotheses. Rigorous theoretical analysis demonstrates that the upper bound of the proposed test statistic converges to zero with high probability under the null hypothesis with the true sender and receiver community numbers \((K_{s}, K_{r})\), while the test statistic itself diverges to infinity under alternatives where the true model exhibits finer community structures. This sharp dichotomy arises from the asymmetric nature of the residual matrix in directed networks, which precludes the direct application of Tracy-Widom distribution theory established for undirected networks—a gap our work addresses by leveraging singular value tail bounds rather than asymptotic distributions for the estimation of the numbers of sender and receiver communities in directed networks. We then propose a sequential testing algorithm DiGoF to lexicographically seek candidate pairs \((k_s, k_r)\) and stop at the smallest pair where \(\hat{T}_n\) falls below a decaying threshold. We further introduce a ratio-based variant algorithm RDiGoF that identifies the transition point in the test statistic sequence to enhance robustness in practical settings. Both algorithms are proven to enjoy consistency in recovering the true \((K_s, K_r)\) within the ScBM framework. Numerical experiments validate the accuracy of both methods across various settings under ScBM, with RDiGoF demonstrating particular advantages in real-world network applications where absolute threshold selection is challenging.

Future research can focus on several interesting questions. First, while we prove the convergence of the upper bound of the test statistic \(\hat{T}_n\) to zero under \(H_0\), establishing whether its lower bound similarly converges to zero would provide a tighter characterization of the statistic’s concentration. Second, even though asymmetry in directed networks causes departures from classical Tracy-Widom limits, investigating whether \(\hat{T}_n\) or its scaled variants follow a novel universal distribution remains an open challenge with profound theoretical implications. Further promising directions include extending this framework to degree-corrected ScBM to accommodate degree heterogeneity \citep{rohe2016co},  developing adaptations for multi-layer ScBM \citep{su2024spectral} or even overlapping community structures \citep{qing2025discovering}, and establishing theoretical guarantees under sparser regimes where edge probabilities decay with network size. These extensions would further enhance the method's applicability to complex real-world directed networks.

\section*{CRediT authorship contribution statement}
\textbf{Huan Qing:} Conceptualization, Data curation, Formal analysis, Funding acquisition, Methodology, Software,
Visualization, Writing – original draft, Writing - review $\&$ editing.

\section*{Declaration of competing interest}
The authors declare no competing interests.

\section*{Data availability}
Data will be made available on request.

\section*{Acknowledgements}
Huan Qing was supported by the Scientific Research Foundation of Chongqing University of Technology (Grant No. 2024ZDR003) and the Science and Technology Research Program of Chongqing Municipal Education Commission (Grant No. KJQN202401168).

\appendix
\section{Proofs for DiGoF}\label{sec:proofs}
Throughout the appendix, we use the notation introduced in the main text. $\|\cdot\|$ denotes the spectral norm, $\|\cdot\|_F$ the Frobenius norm, 
and $\sigma_1(\cdot)$ the largest singular value. For two sequences $a_n,b_n$ we write $a_n\ll b_n$ when $a_n=o(b_n)$. The symbols $C$ and $c$ denote generic positive constants whose values may change from line to line.
\begin{proof}[Proof of Lemma \ref{ideal0}]
By the statements after Corollary 3.12 in \citep{Afonso2016}, we have the following results for tail bounds of rectangular matrices.
\begin{lem}\label{Extension}
Let \(X\) be any \(n \times m\) random matrix whose entries \(X_{ij}\) are independent random variables. Then, for any \(0 < \eta \leq 1/2\), there exists a universal constant \(c_{\eta} > 0\) such that for every \(t \geq 0\),  
\[
\mathbb{P}\left[ \|X\| \geq (1+\eta)\left(\varsigma_1 + \varsigma_2\right) + t \right] \leq (n \wedge m) \exp\left(-\frac{t^2}{c_{\eta} \varsigma_*^2}\right),
\]  
where  
\[
\varsigma_1 = \max_{i} \sqrt{\sum_{j} \mathbb{E}[X_{ij}^2]}, \quad \varsigma_2 = \max_{j} \sqrt{\sum_{i} \mathbb{E}[X_{ij}^2]}, \quad \varsigma_* = \max_{ij} \|X_{ij}\|_{\infty}.
\]
\end{lem}

For $R$, under $H_0$, $A(i,j)$ are independent Bernoulli($\Omega(i,j)$) variables. Thus, we have
\begin{itemize}
\item $\mathbb{E}[R(i,j)] = 0$ for $i \neq j$.
\item $\operatorname{Var}(R(i,j)) = \dfrac{\operatorname{Var}(A(i,j))}{(n-1)\Omega(i,j)(1-\Omega(i,j))} = \dfrac{\Omega(i,j)(1-\Omega(i,j))}{(n-1)\Omega(i,j)(1-\Omega(i,j))} = \dfrac{1}{n-1}$.
\end{itemize}

Thus, let $X$ in Lemma \ref{Extension} be $R$, we have $\varsigma_{1}=1, \varsigma_{2}=1$, and $\varsigma_{*}=\frac{1}{\sqrt{\delta(1-\delta)(n-1)}}$. By Lemma \ref{Extension}, we have
\begin{align*}
\mathbb{P}(\|R\|\geq2(1+\eta)+t)\leq n\exp\left(-\frac{t^{2}\delta(1-\delta)(n-1)}{c_{\eta}}\right).
\end{align*}

Let $t=\sqrt{\frac{\gamma \log (n) c_{\eta}}{\delta(1-\delta)(n-1)}}$ for $\gamma>2$, we get
\begin{align*}
\mathbb{P}(\|R\|\geq2(1+\eta)+t)\leq n\exp\left(\frac{1}{n^{\gamma-1}}\right),
\end{align*}
which means that with probability at least $1-o(\frac{1}{n^{\gamma-1}})$, we have
\begin{align}\label{roughbound1}
\sigma_{1}(R)-2=\|R\|-2\leq2\eta+\sqrt{\frac{\gamma \log (n) c_{\eta}}{\delta(1-\delta)(n-1)}}.
\end{align}

Since $\eta$ is any value in $(0,1/2]$, $\delta, C_\eta$, and  $\gamma$ are positive constants, Equation (\ref{roughbound1}) means that when Assumption \ref{assump:a1} holds, $\mathbb{P}(T_n<\epsilon)\to 1$ for any $\epsilon>0$ as $n\rightarrow\infty$.
\end{proof}
\begin{proof}[Proof of Theorem \ref{thm:null}]
First, we bridge the gap between $R$ and $\hat{R}$. Define $\Delta = \hat{R} - R$. Decompose $\Delta_{ij} = E_{1}(i,j) + E_{2}(i,j)$ where:
\begin{align*}
E_{1}(i,j) &= (A(i,j)- \Omega(i,j)) \left( \frac{1}{\sqrt{\hat{\Omega}(i,j)(1-\hat{\Omega}(i,j))}} - \frac{1}{\sqrt{\Omega(i,j)(1-\Omega(i,j))}} \right) \frac{1}{\sqrt{n-1}} \\
E_{2}(i,j) &= \frac{\Omega(i,j) - \hat{\Omega}(i,j)}{\sqrt{(n-1)\hat{\Omega}(i,j)(1-\hat{\Omega}(i,j))}}.
\end{align*}

By mean value theorem and Assumption \ref{assump:a1}, for $\xi_{ij}$ between $\Omega(i,j)$ and $\hat{\Omega}(i,j)$:
\[
\left| \frac{1}{\sqrt{\hat{\Omega}(i,j)(1-\hat{\Omega}(i,j))}} - \frac{1}{\sqrt{\Omega(i,j)(1-\Omega(i,j))}} \right| = \frac{1}{2} |2\xi_{ij} - 1| [\xi_{ij}(1-\xi_{ij})]^{-3/2} |\hat{\Omega}(i,j) - \Omega(i,j)|.
\]

Since $\hat{\Omega}(i,j)\in [\delta/2, 1-\delta/2]$ w.h.p., by continuity, $|2\xi_{ij} - 1| [\xi_{ij}(1-\xi_{ij})]^{-3/2} \leq C_\delta$ for some $C_\delta > 0$. Thus:
\[
|E_{1}(i,j)| \leq C_\delta |A(i,j) - \Omega(i,j)| \cdot |\hat{\Omega}(i,j) - \Omega(i,j)| \cdot n^{-1/2} \leq C_\delta |\hat{\Omega}(i,j) - \Omega(i,j)| n^{-1/2}
\]
since $|A(i,j) - \Omega(i,j)| \leq 1$. Similarly:
\[
|E_{2}(i,j)| \leq \frac{|\Omega(i,j) - \hat{\Omega}(i,j)|}{\sqrt{(n-1) \cdot (\delta/2)(1-\delta/2)}} = C'_\delta |\hat{\Omega}(i,j) - \Omega(i,j)| n^{-1/2}.
\]

Thus, we have
\[
|\Delta(i,j)| \leq (C_\delta + C'_\delta) |\hat{\Omega}(i,j) - \Omega(i,j)| n^{-1/2}
\]
and
\[
\|\Delta\| \leq \|\Delta\|_F \leq C''_\delta n^{-1/2} \|\hat{\Omega} - \Omega\|_F.
\]

By Lemma \ref{lem:frobenius-error}, we have
\[
\|\hat{\Omega}-\Omega\|_F=O_P\left(K_{\mathrm{max}}\sqrt{\log n}\right).
\]

By Assumption \ref{assump:a3}, we have
\[
\|\Delta\| = O_P\left( n^{-1/2} K_{\mathrm{max}}\sqrt{\log n}\right) =o_P(1).
\]

By Weyl's inequality, we have
\[
|\sigma_1(\hat{R}) - \sigma_1(R)| \leq \|\Delta\| = o_P(1).
\]

Thus, we get
\[
\sigma_1(\hat{R}) = \sigma_1(R) + o_P(1).
\]

By Lemma \ref{ideal0}, we know that $\mathbb{P}(T_n<\epsilon)\to 1$ for any $\epsilon>0$ as $n \to \infty$. Therefore, we have
\[
\hat{T}_n = \sigma_1(\hat{R}) - 2 = \sigma_1(R) - 2 + o_P(1) =T_n+o_{P}(1),
\]
which gives
\[
\mathbb{P}(\hat{T}_n<\epsilon+o_{P}(1))\to 1\Leftrightarrow\mathbb{P}(\hat{T}_n<\epsilon)\to 1.
\]
\end{proof}

\begin{proof}[Proof of Theorem \ref{thm:power}]
Assume $K_s > K_{s0}$ without loss of generality. By the definition of the community separation parameter $\delta$, there exist sender communities $k_s \neq k_r$ and receiver community $l^*$ such that
\[
|B(k_s, l^*) - B(k_r, l^*)| \geq \delta_n.
\]

Define node sets:
\begin{align*}
\mathcal{S}_1 &= \{i: g^s(i) = k_s\}, \quad s_1 = |\mathcal{S}_1| \geq c_0 n / K_s \\
\mathcal{S}_2 &= \{i: g^s(i) = k_r\}, \quad s_2 = |\mathcal{S}_2| \geq c_0 n / K_s \\
\mathcal{T} &= \{j: g^r(j) = l^*\}, \quad t = |\mathcal{T}| \geq c_0 n / K_r.
\end{align*}

Under $H_1$ with $K_{s0} < K_s$, by pigeonhole principle, there exists an estimated sender community $s_0$ such that $\mathcal{S}_1 \cup \mathcal{S}_2 \subseteq \{i: \hat{g}^s(i) = s_0\}$. Meanwhile, since all nodes in \(\mathcal{T}\) share the same true receiver label $l^{*}$ and $K_{r0}\leq K_r$ (possibly $K_{r0}<K_r$), the estimation algorithm forces all nodes in \(\mathcal{T}\) (and potentially other nodes) into one estimated receiver community $\hat{t}$ by pigeonhole principle. Note that when $K_{r0}=K_r$, $\hat{g}^r$ is a consistent receiver community estimator, and the above argument for nodes in \(\mathcal{T}\)  still holds.

\textbf{Step 1: Lower bound for expectation difference.}
Consider the  structured residual submatrix $M = (\Omega - \hat{\Omega})(\mathcal{S}_1 \cup \mathcal{S}_2, \mathcal{T})$. For $i \in \mathcal{S}_1$, $\Omega_{ij} = B(k_s, l^*)$, $\hat{\Omega}(i,j) = \hat{B}(s_0, \hat{t})$; for $i \in \mathcal{S}_2$, $\Omega(i,j) = B(k_r, l^*)$, $\hat{\Omega}(i,j) = \hat{B}(s_0, \hat{t})$. Thus, we have
\[
M = \begin{bmatrix} 
(B(k_s, l^*) - \hat{B}(s_0, \hat{t})) \mathbf{1}_{s_1}^\top \\ 
(B(k_r, l^*) - \hat{B}(s_0, \hat{t})) \mathbf{1}_{s_2}^\top 
\end{bmatrix} \mathbf{1}_t^\top,
\]
where $\mathbf{1}_m^\top$ is a row vector of nodes of length $m$ for any positive integer $m$. Since $M$ is a rank-1 matrix, its spectral norm satisfies
\[
\|M\| = \| \mathbf{v} \mathbf{1}_t^\top \| = \|\mathbf{v}\|_2 \cdot \|\mathbf{1}_t\|_2 = \|\mathbf{v}\|_2 \sqrt{t}
\]
where $\mathbf{v} = \begin{pmatrix} (B(k_s, l^*) - \hat{B}(s_0, \hat{t})) \mathbf{1}_{s_1} \\ (B(k_r, l^*) - \hat{B}(s_0, \hat{t})) \mathbf{1}_{s_2} \end{pmatrix}$ is a column vector of length \(s_1+s_2\). Now, we have
\[
\|\mathbf{v}\|_2^2 = (B(k_s, l^*) - \hat{B}(s_0, \hat{t}))^2 s_1 + (B(k_r, l^*) - \hat{B}(s_0, \hat{t}))^2 s_2 \geq \min(s_1, s_2) \left[ (B(k_s, l^*) - \hat{B}(s_0, \hat{t}))^2 + (B(k_r, l^*) - \hat{B}(s_0, \hat{t}))^2 \right].
\]

Using the inequality $a^2 + b^2 \geq (a - b)^2$ for $a\geq0$ and $b\geq0$ gives
\[
\|\mathbf{v}\|_2^2 \geq \min(s_1, s_2)(B(k_s, l^*) - B(k_r, l^*))^2.
\]

Thus, we get
\[
\|M\| \geq \sqrt{ \min(s_1, s_2) (B(k_s, l^*) - B(k_r, l^*))^2 } \cdot \sqrt{t} = \sqrt{\min(s_1, s_2) t} |B(k_s, l^*) - B(k_r, l^*)|.
\]

Since $\min(s_1, s_2) \geq c_0 n / K_s$ and $t \geq c_0 n / K_r$ by Assumption \ref{assump:a2}, we have
\[
\|M\| \geq \sqrt{ \frac{c_0 n}{K_s} \cdot \frac{c_0 n}{K_r} } \delta_n = c_0 \delta_n \sqrt{ \frac{n^2}{K_s K_r} }.
\]

\textbf{Step 2: Control random fluctuation.}
Let $W = (A - \Omega)(\mathcal{S}_1 \cup \mathcal{S}_2, \mathcal{T})$. We have
\begin{align*}
\sigma_{\text{row}}^2 &= \left\| \sum_{i,j} \mathbb{E} \left[ W^{(ij)} (W^{(ij)})^\top \right] \right\| = \max_i \sum_j \mathrm{Var}(W(i,j)) \leq \frac{t}{4} \leq \frac{n}{4}, \\
\sigma_{\text{col}}^2 &= \left\| \sum_{i,j} \mathbb{E} \left[ (W^{(ij)})^\top W^{(ij)} \right] \right\| = \max_j \sum_i \mathrm{Var}(W(i,j)) \leq \frac{s_1 + s_2}{4} \leq \frac{n}{2},
\end{align*}
where $W^{(ij)}$ is the matrix with single entry $W_{ij}$. Given that $\max(\sigma_{\text{row}}^2, \sigma_{\text{col}}^2) \leq n/2$, by the rectangular case of matrix Bernstein inequality given in Theorem 1.6 of \citep{tropp2012user}, we have
\[
\mathbb{P}\left( \|W\| \geq 3\sqrt{n \log n} \right) \leq (s_1 + s_2 + t) \exp\left( \frac{-(9n \log n)/2}{n/2 + \sqrt{n \log n}} \right) \leq 3n \cdot n^{-9/2}=3n^{-7/2} \to 0,
\]
which gives $\|W\| = O_P(\sqrt{n \log n})$.

\textbf{Step 3: Lower bound for $\sigma_1(\hat{R})$.}
Since $\hat{\Omega}(i,j) \in [\delta/2, 1-\delta/2]$ w.h.p. by Lemma \ref{lem:bound-Omega}, we have
\[
\sigma_1(\hat{R}) \geq \frac{ \| (A - \hat{\Omega})(\mathcal{S}_1 \cup \mathcal{S}_2, \mathcal{T}) \| }{ \max_{i,j} \sqrt{(n-1) \hat{\Omega}(i,j)(1 - \hat{\Omega}(i,j))} } \geq \frac{ \| (A - \hat{\Omega})(\mathcal{S}_1 \cup \mathcal{S}_2, \mathcal{T}) \| }{ \sqrt{n} \cdot 1/2 }
\]
as $\hat{\Omega}(i,j)(1-\hat{\Omega}(i,j)) \leq 1/4$. Now:
\[
\| (A - \hat{\Omega})(\mathcal{S}_1 \cup \mathcal{S}_2, \mathcal{T})\| \geq \|M\| - \|W\| \geq c_0 \delta_n \sqrt{\frac{n^2}{K_s K_r}} - O_P(\sqrt{n \log n}).
\]

Thus, we have
\[
\sigma_1(\hat{R}) \geq \frac{2}{\sqrt{n}} \left( c_0 \delta_n \sqrt{\frac{n^2}{K_s K_r}} - O_P(\sqrt{n \log n}) \right) = 2 c_0 \delta_n \sqrt{\frac{n}{K_s K_r}} - O_P\left(\sqrt{\frac{\log n}{n}}\right).
\]

By Assumption \ref{assump:a3}, we have
\[
\hat{T}_n=\sigma_1(\hat{R}) - 2 \geq 2c_0 \delta_n \sqrt{\frac{n}{K_s K_r}} - 2-O_P\left(\sqrt{\frac{\log n}{n}}\right)\geq 2c_0 \delta_n \sqrt{\frac{n}{K^2_{\mathrm{max}}}}-2 -o_P(1) \to \infty,
\]
with probability $1-O(\frac{1}{n^{7/2}})$.
\end{proof}

\begin{proof}[Proof of Theorem \ref{thm:consistency}]
Let $K^* = K_s + K_r$. Define events:
\begin{align*}
\mathcal{A} &= \{ (\hat{K}_s, \hat{K}_r) \neq (K_s, K_r) \} \\
\mathcal{B} &= \{ \exists (k_s, k_r) <_{\text{lex}} (K_s, K_r) : \hat{T}_n(k_s, k_r) < t_n \} \\
\mathcal{C} &= \{ \hat{T}_n(K_s, K_r) \geq t_n \}.
\end{align*}

Then $\mathcal{A} \subseteq \mathcal{B} \cup \mathcal{C}$, so $\mathbb{P}(\mathcal{A}) \leq \mathbb{P}(\mathcal{B}) + \mathbb{P}(\mathcal{C})$.

\textbf{Part 1: $\mathbb{P}(\mathcal{C}) \to 0$.}
Under $H_0: (k_s, k_r) = (K_s, K_r)$, by Theorem \ref{thm:null}:
\[
\hat{T}_n(K_s, K_r) \xrightarrow{P} 0.
\]

Since $t_n = C n^{-\varepsilon} \to 0$, we get
\[
\mathbb{P}(\hat{T}_n(K_s, K_r) \geq t_n) \to 0.
\]

\textbf{Part 2: $\mathbb{P}(\mathcal{B}) \to 0$.}
The underfitting set is:
\[
\mathcal{S} = \{ (k_s, k_r) \in \mathbb{Z}_+^2 : k_s + k_r \leq K^*, \, (k_s, k_r) \neq (K_s, K_r) \}
\]
with $|\mathcal{S}| \leq (K^*)^2\leq 4K^2_{\mathrm{max}}$. For each $(k_s, k_r) \in \mathcal{S}$, by Theorem \ref{thm:power}:
\[
\hat{T}_n(k_s, k_r) \geq c \sqrt{\frac{n\delta^2_n}{K^2_\mathrm{max}}}.
\]

Since $t_n = n^{-\varepsilon}$ for any $\varepsilon >0$, by Assumption \ref{assump:a3}, we have
\[
\frac{\sqrt{\frac{n\delta^2_n}{K^2_\mathrm{max}}}}{t_n} = \sqrt{\frac{n^{1+2\varepsilon}\delta^2_{n}}{K^2_\mathrm{max}}}\to \infty.
\]

So $\mathbb{P}(\hat{T}_n(k_s, k_r) < t_n) \to 0$. By Assumption \ref{assump:a3} and the proof of Theorem \ref{thm:power}, using union bound gives
\[
\mathbb{P}(\mathcal{B}) \leq |\mathcal{S}| \max_{(k_s,k_r) \in \mathcal{S}} \mathbb{P}(\hat{T}_n(k_s, k_r) < t_n)=O(K^{2}_{\mathrm{max}}\frac{1}{n^{7/2}})= o(1).
\]

\textbf{Part 3: Conclusion.}
Since $\mathbb{P}(\mathcal{C}) \to 0$ and $\mathbb{P}(\mathcal{B}) \to 0$:
\[
\mathbb{P}(\mathcal{A}) \leq \mathbb{P}(\mathcal{B}) + \mathbb{P}(\mathcal{C}) \to 0
\]
which implies:
\[
\lim_{n \to \infty} \mathbb{P}\left( (\hat{K}_s, \hat{K}_r) = (K_s, K_r) \right) = 1.
\]
\end{proof}

\subsection{Useful lemmas}\label{app:frobenius-error}
\begin{lem}\label{lem:frobenius-error}
Under the conditions of Theorem \ref{thm:null}, the estimated expected adjacency matrix \(\hat{\Omega}\) satisfies
\[\|\hat{\Omega} - \Omega\|_F =O_P\left(K_{\mathrm{max}}\sqrt{\log n}\right),
\]
where $m_s$ (and $m_r$) is the number of misclustered nodes in the sending (receiving) side.
\end{lem}

\begin{proof}[Proof of Lemma \ref{lem:frobenius-error}]
Define the misclustering sets
\begin{align*}
\mathcal{M}^s &= \{i: \hat{g}^s(i) \neq g^s(i)\} \mathrm{~with~} m_s = |\mathcal{M}^s|, \\
\mathcal{M}^r &= \{j: \hat{g}^r(j) \neq g^r(j)\} \mathrm{~with~}  m_r = |\mathcal{M}^r|,
\end{align*}

The error is decomposed as:
\[
\hat{\Omega}(i,j) - \Omega(i,j) = \underbrace{\left( \hat{B}(k,l) - \bar{B}(k,l) \right)}_{\mathbb{V}_{ij}} + \underbrace{\left( \bar{B}(k,l) - B(g^s(i), g^r(j)) \right)}_{\mathbb{B}_{ij}},
\]
for \(k = \hat{g}^s(i)\), \(l = \hat{g}^r(j)\), where \(\bar{B}(k,l)\) is the conditional expectation:
\[
\bar{B}(k,l) = \mathbb{E}[\hat{B}(k,l) \mid \hat{g}^s, \hat{g}^r] = \frac{1}{|\hat{C}_k^s| \cdot |\hat{C}_l^r|} \sum_{i' \in \hat{C}_k^s} \sum_{j' \in \hat{C}_l^r} B(g^s(i'), g^r(j')).
\]

The squared Frobenius norm is:
\[
\|\hat{\Omega} - \Omega\|_F^2 = \sum_{i \neq j} (\hat{\Omega}(i,j) - \Omega(i,j))^2 = \underbrace{\sum_{i \neq j} \mathbb{V}_{ij}^2}_{\mathbb{V}} + \underbrace{\sum_{i \neq j} \mathbb{B}_{ij}^2}_{\mathbb{B}} + \underbrace{2\sum_{i \neq j} \mathbb{V}_{ij} \mathbb{B}_{ij}}_{\mathbb{C}}.
\]

For fixed estimated communities, \(\hat{B}(k,l)\) is the average of independent Bernoulli variables. By Lemma \ref{LemhatBEhatBDiff}, we have
\[
|\hat{B}(k,l) - \bar{B}(k,l)| = O_P\left( \sqrt{\frac{\log n}{|\hat{C}_k^s| |\hat{C}_l^r|}} \right) = O_P\left( \frac{K_{\mathrm{max}}\sqrt{\log n}}{n} \right),
\]
which gives
\[
\mathbb{V} = \sum_{k,l} \sum_{i \in \hat{C}_k^s} \sum_{j \in \hat{C}_l^r} \mathbb{V}_{ij}^2 = \sum_{k,l} |\hat{C}_k^s| |\hat{C}_l^r| \cdot O_P\left( \frac{K^2_{\mathrm{max}}\log n}{n^2} \right) = O_P(K^2_{\mathrm{max}} \log n).
\]

The bias term \(B_{ij} = \bar{B}(k,l) - B(g^s(i), g^r(j))\) is bounded by considering node misclustering:
\[
|B_{ij}| \leq \frac{1}{|\hat{C}_k^s| |\hat{C}_l^r|} \sum_{i' \in \hat{C}_k^s} \sum_{j' \in \hat{C}_l^r} \left| B(g^s(i'), g^r(j')) - B(g^s(i), g^r(j)) \right|.
\]

Define:
\begin{align*}
\mathcal{G}_k^s &= \{i' \in \hat{C}_k^s: g^s(i') = k\}, \quad \mathcal{B}_k^s = \hat{C}_k^s \setminus \mathcal{G}_k^s, \quad |\mathcal{B}_k^s| \leq m_s, \\
\mathcal{G}_l^r &= \{j' \in \hat{C}_l^r: g^r(j') = l\}, \quad \mathcal{B}_l^r = \hat{C}_l^r \setminus \mathcal{G}_l^r, \quad |\mathcal{B}_l^r| \leq m_r.
\end{align*}

Then:
\[
|B_{ij}| \leq \frac{|\mathcal{B}_k^s| \cdot |\hat{C}_l^r| + |\mathcal{B}_l^r| \cdot |\hat{C}_k^s|}{|\hat{C}_k^s| |\hat{C}_l^r|} \leq \frac{m_s}{|\hat{C}_k^s|} + \frac{m_r}{|\hat{C}_l^r|} \leq \frac{K_{\mathrm{max}}}{c_0 n} (m_s + m_r).
\]

The sum \(\mathbb{B}\) is split by clustering status:
\[
\mathbb{B} = \sum_{\substack{i \notin \mathcal{M}^s \\ j \notin \mathcal{M}^r}} \mathbb{B}_{ij}^2 + \sum_{\substack{i \in \mathcal{M}^s \text{ or } j \in \mathcal{M}^r}} \mathbb{B}_{ij}^2.
\]

For correctly clustered nodes (\(i \notin \mathcal{M}^s\), \(j \notin \mathcal{M}^r\)):
\[
|\mathbb{B}_{ij}| \leq \frac{K_{\mathrm{max}}}{c_0 n} (m_s + m_r) \implies \sum_{\substack{i \notin \mathcal{M}^s \\ j \notin \mathcal{M}^r}} \mathbb{B}_{ij}^2 \leq \left( \frac{K_{\mathrm{max}}}{c_0 n} (m_s + m_r) \right)^2 \cdot n^2 = O_P\left(K_{\mathrm{max}}^2 (m_s + m_r)^2 \right).
\]

For misclustered nodes (at least one of \(i\) or \(j\) misclustered), since \(|\mathbb{B}_{ij}| \leq 1\), we have
\[
\sum_{\substack{i \in \mathcal{M}^s \text{ or } j \in \mathcal{M}^r}} \mathbb{B}_{ij}^2\leq n(m_{r}+m_{s}).
\]

Thus, we get
\[
\mathbb{B}=O_P\left(K_{\mathrm{max}}^2 (m_s + m_r)^2+ n(m_{s}+m_{r})\right).
\]

By the Cauchy-Schwarz inequality:
\[
|\mathbb{C}| \leq 2 \sqrt{\mathbb{V}} \sqrt{\mathbb{B}} = 2 \sqrt{O_P(K_{\mathrm{max}}^2 \log n)} \cdot \sqrt{ O_P\left(K_{\mathrm{max}}^2 (m_s + m_r)^2 + n(m_s+ m_r)) \right)}.
\]

Given that $\hat{g}^{s}$ and $\hat{g}^{r}$ are consistent community estimators under the null hypothesis $H_0$, we have $\mathbb{P}(m_{s}=0)\rightarrow1$ and $\mathbb{P}(m_{r}=0)\rightarrow1$.Thus, we have
\[\|\hat{\Omega} - \Omega\|_F =O_P\left( K_{\mathrm{max}}\sqrt{\log}\right).
\]

\end{proof}

\begin{lem}\label{LemhatBEhatBDiff}
Under Assumption \ref{assump:a1}, for any estimated sender community \(s_0\) and receiver community \(\hat{t}\) with sizes \(n_s = |\{i: \hat{g}^s(i) = s_0\}|\) and \(n_t = |\{j: \hat{g}^r(j) = \hat{t}\}|\), the estimated block probability \(\hat{B}(s_0, \hat{t})\) satisfies:
\[
|\hat{B}(s_0, \hat{t}) - \mathbb{E}[\hat{B}(s_0, \hat{t})]| \leq \sqrt{\frac{3 \log n}{n_s n_t}}
\]
with probability at least \(1 - O(n^{-3})\).
\end{lem}
\begin{proof}[Proof of Lemma \ref{LemhatBEhatBDiff}]
The estimated block probability is defined as
\[
\hat{B}(s_0, \hat{t}) = \frac{1}{n_s n_t} \sum_{i \in \hat{C}_{s_0}^s} \sum_{j \in \hat{C}_{\hat{t}}^r} A(i,j),
\]
where \(\hat{C}_{s_0}^s = \{i: \hat{g}^s(i) = s_0\}\) and \(\hat{C}_{\hat{t}}^r = \{j: \hat{g}^r(j) = \hat{t}\}\). Conditional on the estimated communities \(\hat{g}^s\) and \(\hat{g}^r\), the random variables \(\{A(i,j)\}_{i \in \hat{C}_{s_0}^s, j \in \hat{C}_{\hat{t}}^r}\) are independent Bernoulli trials with success probabilities \(p_{ij} = B(g^s(i), g^r(j)) \in [\delta, 1-\delta]\) by Assumption \ref{assump:a1}. Define the centered variables:
\[
X_{ij} = A(i,j) - p_{ij}, \quad \text{for} \quad i \in \hat{C}_{s_0}^s,  j \in \hat{C}_{\hat{t}}^r.
\]

These satisfy \(\mathbb{E}[X_{ij} \mid \hat{g}^s, \hat{g}^r] = 0\), \(|X_{ij}| \leq 1\), and \(\operatorname{Var}(X_{ij} \mid \hat{g}^s, \hat{g}^r) = p_{ij}(1-p_{ij}) \leq \frac{1}{4}\). Consider the sum 
\[
S = \sum_{i \in \hat{C}_{s_0}^s} \sum_{j \in \hat{C}_{\hat{t}}^r} X_{ij}.
\]

The variance of \(S\) is bounded by
\[
V = \operatorname{Var}(S \mid \hat{g}^s, \hat{g}^r) \leq \sum_{i,j} \frac{1}{4} = \frac{n_s n_t}{4}.
\]

For any \(v > 0\), applying Bernstein's inequality \citep{tropp2012user} to \(S\) gives
\[
\mathbb{P}\left( |S| \geq v \mid \hat{g}^s, \hat{g}^r \right) \leq 2 \exp\left( -\frac{v^2/2}{V + v/3} \right).
\]

Substitute \(v = \sqrt{3 n_s n_t \log n}\) and \(V \leq \frac{n_s n_t}{4}\):
\[
\mathbb{P}\left( |S| \geq \sqrt{3 n_s n_t \log n} \mid \hat{g}^s, \hat{g}^r \right) \leq 2 \exp\left( -\frac{3 n_s n_t \log n / 2}{\frac{n_s n_t}{4} + \frac{\sqrt{3 n_s n_t \log n}}{3}} \right).
\]

For large \(n\), \(\sqrt{3 n_s n_t \log n}/3 \leq n_s n_t/4\) since \(n_s n_t \gg \sqrt{n_s n_t \log n}\). Thus, we have
\[
\frac{3 n_s n_t \log n / 2}{n_s n_t / 2} = 3 \log n, \quad \text{so} \quad \mathbb{P}\left( |S| \geq \sqrt{3 n_s n_t \log n} \mid \hat{g}^s, \hat{g}^r \right) \leq 2 \exp(-3 \log n) = 2n^{-3}.
\]

Finally, with probability at least \(1-O(n^{-3})\), we have
\[
|\hat{B}(s_0, \hat{t}) - \mathbb{E}[\hat{B}(s_0, \hat{t}) \mid \hat{g}^s, \hat{g}^r]| = \frac{|S|}{n_s n_t} \leq \sqrt{\frac{3 \log n}{n_s n_t}}.
\]
\end{proof}

\begin{lem}\label{lem:bound-Omega}
Under Assumptions \ref{assump:a1}-\ref{assump:a3}, and under the alternative hypothesis ($K_{s0} < K_s$ or $K_{r0} < K_r$), the estimated probability matrix entries satisfy  
\[
\hat{\Omega}(i,j) \in [\delta/2, 1-\delta/2] \quad \text{for all } i,j
\]  
with probability at least $1 - o(n^{-1})$.  
\end{lem}

\begin{proof}[Proof of Lemma \ref{lem:bound-Omega}]
Recall that for any $i,j$, $\hat{\Omega}(i,j) = \hat{B}(k,l)$ where $k = \hat{g}^s(i)$, $l = \hat{g}^r(j)$. We prove $\hat{B}(k,l) \in [\delta/2, 1-\delta/2]$ uniformly over all blocks $k,l$. The conditional expectation  
\[
\bar{B}(k,l) = \mathbb{E}[\hat{B}(k,l) \mid \hat{g}^s, \hat{g}^r] = \frac{1}{|\hat{C}_k^s||\hat{C}_l^r|} \sum_{i' \in \hat{C}_k^s} \sum_{j' \in \hat{C}_l^r} B(g^s(i'), g^r(j'))
\]  
is a convex combination of $B(\cdot,\cdot) \in [\delta, 1-\delta]$ (Assumption \ref{assump:a1}). Thus,  
\[
\bar{B}(k,l) \in [\delta, 1-\delta] \quad \forall k,l.
\]

By Lemma \ref{LemhatBEhatBDiff}, for any fixed $(k,l)$, with probability at least $1-O(\frac{1}{n^3})$, we have   
\[
|\hat{B}(k,l) - \bar{B}(k,l)| \leq \sqrt{\frac{3 \log n}{n_k n_l}},
\]  
where $n_k = |\hat{C}_k^s|$, $n_l = |\hat{C}_l^r|$. Under $H_1$ ($K_{s0} \leq K_s$, $K_{r0} \leq K_r$), the pigeonhole principle implies each estimated community contains at least one true community. By Assumption \ref{assump:a2}, we have  
\[
n_k \geq \min_{k'} |\{i: g^s(i) = k'\}| \geq c_0 \frac{n}{K_s} \geq c_0 \frac{n}{K_{\max}}, \quad 
n_l \geq c_0 \frac{n}{K_{\max}}.
\]  

Thus, $n_k n_l \geq \left(c_0 n / K_{\max}\right)^2$, which gives
\[
\sqrt{\frac{3 \log n}{n_k n_l}} \leq \frac{K_{\max} \sqrt{3 \log n}}{c_0 n}.
\]  

By Assumption \ref{assump:a3}, we have 
\[
\frac{K_{\max} \sqrt{\log n}}{n}=\frac{1}{\sqrt{n}}\sqrt{\frac{K^{2}_{\mathrm{max}}\log n}{n}} = o(\frac{1}{\sqrt{n}}).
\]  

For large $n$, $\frac{1}{\sqrt{n}}\ll\delta/2$. Thus, with probability at least $1 - O(n^{-3})$, for per block, we have  
\[
\hat{B}(k,l) \in \left[\bar{B}(k,l) - \frac{\delta}{2}, \bar{B}(k,l) + \frac{\delta}{2}\right] \subseteq \left[\frac{\delta}{2}, 1 - \frac{\delta}{2}\right].
\]  

There are $\leq K_{s0}K_{r0} \leq K_{\max}^2$ blocks. Then, by Assumption \ref{assump:a3}, the union bound gives 
\[
\mathbb{P}\left(\bigcup_{k,l} \left\{ |\hat{B}(k,l) - \bar{B}(k,l)| > \frac{\delta}{2} \right\}\right) = O(K_{\max}^2 n^{-3})\ll O(\frac{n}{\log n}n^{-3}) = O(\frac{n^{-2}}{\log n}).
\]  

Thus, with probability at least $1 - o(n^{-1})$, we have all $\hat{\Omega}(i,j) \in [\delta/2, 1-\delta/2]$.  
\end{proof}

\section{Proofs for RDiGoF}\label{sec:proofs-RDiGoF}
This section provides the detailed proofs of Theorems \ref{thm:ratio} and \ref{thm:RDiGoF-consistency}. We begin by establishing several useful lemmas.

\begin{lem}\label{lem:Tn-behavior}
Under the assumptions of Theorem \ref{thm:ratio}, we have
\begin{enumerate}
    \item $\hat{T}_n(m) \geq -2$.
    \item For underfitted models ($m < m_*$): For any $c > 0$, $\lim_{n \to \infty} \mathbb{P}(\hat{T}_n(m) > c) = 1$.
    \item For the true model ($m = m_*$): $\lim_{n \to \infty} \mathbb{P}(|\hat{T}_n(m_*)| \leq 2) = 1$. 
\end{enumerate}
\end{lem}

\begin{proof}[Proof of Lemma \ref{lem:Tn-behavior}]
Since singular values are non-negative, $\sigma_1(\hat{R}) \geq 0$, so $\hat{T}_n(m) = \sigma_1(\hat{R}) - 2 \geq -2$. 

Let $d_n = \frac{\delta_n \sqrt{n}}{K_{\max}}$. From the proof of Theorem \ref{thm:power}, there exists $c_1> 0$ such that $
\hat{T}_n(m) \geq c_1 d_n - 2 - o_P(1).$ Since $K_{\mathrm{max}}=O(1)$ and $\delta_{n}\geq\delta_{\mathrm{min}}$, $d_n = \delta_n \sqrt{n} / K_{\max} \geq \delta_{\min} \sqrt{n} / K_{\max} =O(\sqrt{n})\to \infty$. Then there exists $N$ such that for $n > N$, $\mathbb{P}(\hat{T}_n(m) > c) = 1$ for any $c> 0$. 
Hence, $\lim_{n \to \infty} \mathbb{P}(\hat{T}_n(m) > c) = 1$.

By Theorem \ref{thm:null}, for any $\epsilon > 0$, $\lim_{n \to \infty} \mathbb{P}(\hat{T}_n(m_*) < \epsilon) = 1$. Taking $\epsilon = 2$, we have $\mathbb{P}(\hat{T}_n(m_*) < 2) \to 1$. Since $\hat{T}_n(m_*) \geq -2$ always holds, it follows that $\mathbb{P}(|\hat{T}_n(m_*)| \leq 2) \to 1$.
\end{proof}

\begin{lem}\label{lem:ratio-bound}
Under the assumptions of Theorem \ref{thm:ratio}, there exists a constant $C > 0$ such that for all $m < m_*$, we have
\[
\lim_{n \to \infty} \mathbb{P}(r_m \leq C) = 1.
\]
\end{lem}

\begin{proof}[Proof of Lemma \ref{lem:ratio-bound}]
From the proof of Theorem \ref{thm:power},
\[
\hat{T}_n(m) \geq c_1 d_n - 2 - o_P(1).
\]
Take $c_{\text{low}} = c_1/2$. Since $d_n \to \infty$, there exists $N$ such that for $n > N$, $c_1 d_n - 2-o_{P}(1)> c_{\text{low}} d_n + 1$. Hence, we have $\lim_{n \to \infty} \mathbb{P}\left( \hat{T}_n(m) \geq c_{\text{low}} d_n \right) = 1$.

By Lemma \ref{lem:bound-Omega}, with probability at least $1 - o(n^{-1})$, $\hat{\Omega}(i,j) \in [\delta/2, 1-\delta/2]$ for all $i,j$. Therefore, we have
\[
\hat{\Omega}(i,j)(1 - \hat{\Omega}(i,j)) \geq \gamma_\delta = (\delta/2)(1-\delta/2) \geq \frac{\delta(1-\delta)}{4}.
\]

Let $\mu = \delta(1-\delta)/4$. Then for $i \neq j$,
\[
|\hat{R}(i,j)| = \frac{|A(i,j) - \hat{\Omega}(i,j)|}{\sqrt{(n-1) \hat{\Omega}(i,j)(1-\hat{\Omega}(i,j))}} \leq \frac{|A(i,j) - \hat{\Omega}(i,j)|}{\sqrt{(n-1) \mu}}.
\]

Thus, the spectral norm satisfies
\[
\| \hat{R} \| \leq \frac{1}{\sqrt{(n-1) \mu}} \| A - \hat{\Omega} \|.
\]

Now, we bound $\|A-\hat{\Omega}\|$ by decomposing it as $A - \hat{\Omega} = (A - \Omega) + (\Omega - \hat{\Omega})$. Since $A - \Omega$ has independent mean-zero entries with variance at most $1/4$, by Bernstein’s inequality, we have $\| A - \Omega \| = O_P(\sqrt{n})$. Since entries of $\Omega$ and $\hat{\Omega}$ are in $[0,1]$, we have $\| \Omega - \hat{\Omega} \|_F \leq n$, hence $\| \Omega - \hat{\Omega} \| \leq \| \Omega - \hat{\Omega} \|_F \leq n$. Therefore, we have
    \[
    \| A - \hat{\Omega} \| \leq \| A - \Omega \| + \| \Omega - \hat{\Omega} \| = O_P(\sqrt{n}) + n = O(n),
    \]
which gives
  \[
  \hat{T}_n(m)=\| \hat{R} \|-2\leq \frac{O(n)}{\sqrt{(n-1) \mu}}-2= O(\sqrt{n}).
  \]

Since $\delta_n \geq \delta_{\min} > 0$ and $K_{\max}$ is bounded, we have $d_n = \delta_n \sqrt{n} / K_{\max} = O(\sqrt{n})$. Therefore, there exists a constant $C_{\text{upp}} > 0$ such that
  \[
  \lim_{n \to \infty} \mathbb{P}\left( \hat{T}_n(m) \leq C_{\text{upp}} d_n \right) = 1.
  \]

Therefore, we get
\[
r_m = \frac{\hat{T}_n(m-1)}{\hat{T}_n(m)} \leq \frac{C_{\text{upp}} d_n}{c_{\text{low}} d_n} = \frac{C_{\text{upp}}}{c_{\text{low}}}=C.
\]
\end{proof}
\begin{proof}[Proof of Theorem \ref{thm:ratio}]
For underfitted models ($m < m_*$): By Lemma \ref{lem:ratio-bound}, $r_m \leq C$ with high probability, so $r_m = O_P(1)$.

For the true model ($m = m_*$): Let $X = \hat{T}_n(m_*-1)$ and $Y = \hat{T}_n(m_*)$. By Lemma \ref{lem:Tn-behavior}, we have
\[
\lim_{n \to \infty} \mathbb{P}(X > 2c) = 1 \quad \text{for any} \quad c> 0, \quad \text{and} \quad \lim_{n \to \infty} \mathbb{P}(|Y| \leq 2) = 1.
\]

Fix $c> 0$. Then, we get
\[
\mathbb{P}(X > 2c) \to 1, \quad \mathbb{P}(|Y| \leq 2) \to 1.
\]

Thus, the ratio statistic $r_{m_*} = \left| \frac{X}{Y} \right|\geq \frac{X}{2}>c$ with high probability. Therefore, for any $c>0$, we have
\[
\mathbb{P}(r_{m_*} > c)\to 1.
\]

Hence, we have $r_{m_*} \stackrel{P}{\to} \infty$.
\end{proof}

\begin{proof}[Proof of Theorem \ref{thm:RDiGoF-consistency}]
Fix $\tau > C$, where $C$ is the one in Lemma \ref{lem:ratio-bound}. Define two events:
\[
E_1 = \{ r_{m_*} > 2\tau \}, \quad E_2 = \{ r_m \leq \tau \text{ for all } m = 2, \dots, m_*-1 \}.
\]

By Theorem \ref{thm:ratio}, for any $\tau' > 0$, $\mathbb{P}(r_{m_*} > \tau') \to 1$. In particular, take $\tau' = 2\tau$, then $\mathbb{P}(E_1) \to 1$. By Lemma \ref{lem:ratio-bound}, for each $m < m_*$, $\mathbb{P}(r_m > \tau) \to 0$. Since $m_* \leq K_{\max}^2$ and $K_{\max}$ is fixed, the number of terms is finite. Thus, by the union bound, we get
\[
\mathbb{P}(E_2^c) \leq \sum_{m=2}^{m_*-1} \mathbb{P}(r_m > \tau) \to 0.
\]

Hence, $\mathbb{P}(E_2) \to 1$. On $E_1 \cap E_2$, we have $r_{m_*} > 2\tau > \tau$ and for all $m < m_*$, $r_m \leq \tau < r_{m_*}$, so $r_{m_*}$ is the first value greater than $\tau$ and is the maximum up to $m_*$. Therefore, the RDiGoF algorithm sets $\hat{m}_* = m_*$, where $\hat{m}_*$ is the first $m$ satisfying $r_{m}>\tau$ in the RDiGoF algorithm. Thus, we have
\[
\mathbb{P}(\hat{m}_* = m_*) \geq \mathbb{P}(E_1 \cap E_2) \to 1,
\]
which implies $\lim_{n \to \infty} \mathbb{P}\left( (\hat{K}_s, \hat{K}_r) = (K_s, K_r) \right) = 1$.
\end{proof}
\bibliographystyle{elsarticle-harv}
\bibliography{reference}

\begin{thebibliography}{40}
\expandafter\ifx\csname natexlab\endcsname\relax\def\natexlab#1{#1}\fi
\providecommand{\url}[1]{\texttt{#1}}
\providecommand{\href}[2]{#2}
\providecommand{\path}[1]{#1}
\providecommand{\DOIprefix}{doi:}
\providecommand{\ArXivprefix}{arXiv:}
\providecommand{\URLprefix}{URL: }
\providecommand{\Pubmedprefix}{pmid:}
\providecommand{\doi}[1]{\href{http://dx.doi.org/#1}{\path{#1}}}
\providecommand{\Pubmed}[1]{\href{pmid:#1}{\path{#1}}}
\providecommand{\bibinfo}[2]{#2}
\ifx\xfnm\relax \def\xfnm[#1]{\unskip,\space#1}\fi
\bibitem[{Bandeira and van Handel(2016)}]{Afonso2016}
\bibinfo{author}{Bandeira, A.S.}, \bibinfo{author}{van Handel, R.},
  \bibinfo{year}{2016}.
\newblock \bibinfo{title}{{Sharp nonasymptotic bounds on the norm of random
  matrices with independent entries}}.
\newblock \bibinfo{journal}{Annals of Probability} \bibinfo{volume}{44},
  \bibinfo{pages}{2479 -- 2506}.
\bibitem[{Bickel and Sarkar(2016)}]{bickel2016hypothesis}
\bibinfo{author}{Bickel, P.J.}, \bibinfo{author}{Sarkar, P.},
  \bibinfo{year}{2016}.
\newblock \bibinfo{title}{Hypothesis testing for automated community detection
  in networks}.
\newblock \bibinfo{journal}{Journal of the Royal Statistical Society Series B:
  Statistical Methodology} \bibinfo{volume}{78}, \bibinfo{pages}{253--273}.
\bibitem[{Chen and Lei(2018)}]{chen2018network}
\bibinfo{author}{Chen, K.}, \bibinfo{author}{Lei, J.}, \bibinfo{year}{2018}.
\newblock \bibinfo{title}{Network cross-validation for determining the number
  of communities in network data}.
\newblock \bibinfo{journal}{Journal of the American Statistical Association}
  \bibinfo{volume}{113}, \bibinfo{pages}{241--251}.
\bibitem[{Coleman et~al.(1957)Coleman, Katz and Menzel}]{coleman1957diffusion}
\bibinfo{author}{Coleman, J.}, \bibinfo{author}{Katz, E.},
  \bibinfo{author}{Menzel, H.}, \bibinfo{year}{1957}.
\newblock \bibinfo{title}{The diffusion of an innovation among physicians}.
\newblock \bibinfo{journal}{Sociometry} \bibinfo{volume}{20},
  \bibinfo{pages}{253--270}.
\bibitem[{Dong et~al.(2020)Dong, Wang and Liu}]{dong2020spectral}
\bibinfo{author}{Dong, Z.}, \bibinfo{author}{Wang, S.}, \bibinfo{author}{Liu,
  Q.}, \bibinfo{year}{2020}.
\newblock \bibinfo{title}{Spectral based hypothesis testing for community
  detection in complex networks}.
\newblock \bibinfo{journal}{Information Sciences} \bibinfo{volume}{512},
  \bibinfo{pages}{1360--1371}.
\bibitem[{Guhl(1953)}]{guhl1953social}
\bibinfo{author}{Guhl, A.M.}, \bibinfo{year}{1953}.
\newblock \bibinfo{title}{Social behavior of the domestic fowl.} .
\bibitem[{Guo et~al.(2016)Guo, Zhang and Yorke-Smith}]{guo2016novel}
\bibinfo{author}{Guo, G.}, \bibinfo{author}{Zhang, J.},
  \bibinfo{author}{Yorke-Smith, N.}, \bibinfo{year}{2016}.
\newblock \bibinfo{title}{A novel evidence-based bayesian similarity measure
  for recommender systems}.
\newblock \bibinfo{journal}{ACM Transactions on the Web (TWEB)}
  \bibinfo{volume}{10}, \bibinfo{pages}{1--30}.
\bibitem[{Guo et~al.(2023)Guo, Qiu, Zhang and Chang}]{guo2023randomized}
\bibinfo{author}{Guo, X.}, \bibinfo{author}{Qiu, Y.}, \bibinfo{author}{Zhang,
  H.}, \bibinfo{author}{Chang, X.}, \bibinfo{year}{2023}.
\newblock \bibinfo{title}{Randomized spectral co-clustering for large-scale
  directed networks}.
\newblock \bibinfo{journal}{Journal of Machine Learning Research}
  \bibinfo{volume}{24}, \bibinfo{pages}{1--68}.
\bibitem[{Holland et~al.(1983)Holland, Laskey and
  Leinhardt}]{holland1983stochastic}
\bibinfo{author}{Holland, P.W.}, \bibinfo{author}{Laskey, K.B.},
  \bibinfo{author}{Leinhardt, S.}, \bibinfo{year}{1983}.
\newblock \bibinfo{title}{{Stochastic blockmodels: First steps}}.
\newblock \bibinfo{journal}{Social Networks} \bibinfo{volume}{5},
  \bibinfo{pages}{109--137}.
\bibitem[{Hu et~al.(2020)Hu, Qin, Yan and Zhao}]{hu2020corrected}
\bibinfo{author}{Hu, J.}, \bibinfo{author}{Qin, H.}, \bibinfo{author}{Yan, T.},
  \bibinfo{author}{Zhao, Y.}, \bibinfo{year}{2020}.
\newblock \bibinfo{title}{Corrected bayesian information criterion for
  stochastic block models}.
\newblock \bibinfo{journal}{Journal of the American Statistical Association}
  \bibinfo{volume}{115}, \bibinfo{pages}{1771--1783}.
\bibitem[{Hu et~al.(2021)Hu, Zhang, Qin, Yan and Zhu}]{hu2021using}
\bibinfo{author}{Hu, J.}, \bibinfo{author}{Zhang, J.}, \bibinfo{author}{Qin,
  H.}, \bibinfo{author}{Yan, T.}, \bibinfo{author}{Zhu, J.},
  \bibinfo{year}{2021}.
\newblock \bibinfo{title}{Using maximum entry-wise deviation to test the
  goodness of fit for stochastic block models}.
\newblock \bibinfo{journal}{Journal of the American Statistical Association}
  \bibinfo{volume}{116}, \bibinfo{pages}{1373--1382}.
\bibitem[{Hwang et~al.(2024)Hwang, Xu, Chatterjee and
  Bhattacharyya}]{hwang2024estimation}
\bibinfo{author}{Hwang, N.}, \bibinfo{author}{Xu, J.},
  \bibinfo{author}{Chatterjee, S.}, \bibinfo{author}{Bhattacharyya, S.},
  \bibinfo{year}{2024}.
\newblock \bibinfo{title}{On the estimation of the number of communities for
  sparse networks}.
\newblock \bibinfo{journal}{Journal of the American Statistical Association}
  \bibinfo{volume}{119}, \bibinfo{pages}{1895--1910}.
\bibitem[{Ji and Jin(2016)}]{PengshengAOAS896}
\bibinfo{author}{Ji, P.}, \bibinfo{author}{Jin, J.}, \bibinfo{year}{2016}.
\newblock \bibinfo{title}{{Coauthorship and citation networks for
  statisticians}}.
\newblock \bibinfo{journal}{Annals of Applied Statistics} \bibinfo{volume}{10},
  \bibinfo{pages}{1779 -- 1812}.
\bibitem[{Ji et~al.(2022)Ji, Jin, Ke and Li}]{ji2022co}
\bibinfo{author}{Ji, P.}, \bibinfo{author}{Jin, J.}, \bibinfo{author}{Ke,
  Z.T.}, \bibinfo{author}{Li, W.}, \bibinfo{year}{2022}.
\newblock \bibinfo{title}{Co-citation and co-authorship networks of
  statisticians}.
\newblock \bibinfo{journal}{Journal of Business \& Economic Statistics}
  \bibinfo{volume}{40}, \bibinfo{pages}{469--485}.
\bibitem[{Jin et~al.(2023)Jin, Ke, Luo and Wang}]{jin2023optimal}
\bibinfo{author}{Jin, J.}, \bibinfo{author}{Ke, Z.T.}, \bibinfo{author}{Luo,
  S.}, \bibinfo{author}{Wang, M.}, \bibinfo{year}{2023}.
\newblock \bibinfo{title}{Optimal estimation of the number of network
  communities}.
\newblock \bibinfo{journal}{Journal of the American Statistical Association}
  \bibinfo{volume}{118}, \bibinfo{pages}{2101--2116}.
\bibitem[{Karrer and Newman(2011)}]{karrer2011stochastic}
\bibinfo{author}{Karrer, B.}, \bibinfo{author}{Newman, M.E.},
  \bibinfo{year}{2011}.
\newblock \bibinfo{title}{Stochastic blockmodels and community structure in
  networks}.
\newblock \bibinfo{journal}{Physical Review E—Statistical, Nonlinear, and
  Soft Matter Physics} \bibinfo{volume}{83}, \bibinfo{pages}{016107}.
\bibitem[{Kunegis(2013)}]{kunegis2013konect}
\bibinfo{author}{Kunegis, J.}, \bibinfo{year}{2013}.
\newblock \bibinfo{title}{Konect: the koblenz network collection}, in:
  \bibinfo{booktitle}{Proceedings of the 22nd International Conference on World
  Wide Web}, pp. \bibinfo{pages}{1343--1350}.
\bibitem[{Le and Levina(2022)}]{CanMK2022}
\bibinfo{author}{Le, C.M.}, \bibinfo{author}{Levina, E.}, \bibinfo{year}{2022}.
\newblock \bibinfo{title}{{Estimating the number of communities by spectral
  methods}}.
\newblock \bibinfo{journal}{Electronic Journal of Statistics}
  \bibinfo{volume}{16}, \bibinfo{pages}{3315 -- 3342}.
\bibitem[{Lei(2016)}]{lei2016goodness}
\bibinfo{author}{Lei, J.}, \bibinfo{year}{2016}.
\newblock \bibinfo{title}{A goodness-of-fit test for stochastic block models}.
\newblock \bibinfo{journal}{Annals of Statistics} \bibinfo{volume}{44},
  \bibinfo{pages}{401--424}.
\bibitem[{{Lei} and {Rinaldo}(2015)}]{lei2015consistency}
\bibinfo{author}{{Lei}, J.}, \bibinfo{author}{{Rinaldo}, A.},
  \bibinfo{year}{2015}.
\newblock \bibinfo{title}{Consistency of spectral clustering in stochastic
  block models}.
\newblock \bibinfo{journal}{Annals of Statistics} \bibinfo{volume}{43},
  \bibinfo{pages}{215--237}.
\bibitem[{Leicht and Newman(2008)}]{leicht2008community}
\bibinfo{author}{Leicht, E.A.}, \bibinfo{author}{Newman, M.E.},
  \bibinfo{year}{2008}.
\newblock \bibinfo{title}{Community structure in directed networks}.
\newblock \bibinfo{journal}{Physical Review Letters} \bibinfo{volume}{100},
  \bibinfo{pages}{118703}.
\bibitem[{Li and Yang(2022)}]{li2022undirected}
\bibinfo{author}{Li, B.}, \bibinfo{author}{Yang, Y.}, \bibinfo{year}{2022}.
\newblock \bibinfo{title}{{Undirected and directed network analysis of the
  Chinese stock market}}.
\newblock \bibinfo{journal}{Computational Economics} \bibinfo{volume}{60},
  \bibinfo{pages}{1155--1173}.
\bibitem[{Li et~al.(2020)Li, Levina and Zhu}]{li2020network}
\bibinfo{author}{Li, T.}, \bibinfo{author}{Levina, E.}, \bibinfo{author}{Zhu,
  J.}, \bibinfo{year}{2020}.
\newblock \bibinfo{title}{Network cross-validation by edge sampling}.
\newblock \bibinfo{journal}{Biometrika} \bibinfo{volume}{107},
  \bibinfo{pages}{257--276}.
\bibitem[{Ma et~al.(2021)Ma, Su and Zhang}]{ma2021determining}
\bibinfo{author}{Ma, S.}, \bibinfo{author}{Su, L.}, \bibinfo{author}{Zhang,
  Y.}, \bibinfo{year}{2021}.
\newblock \bibinfo{title}{Determining the number of communities in
  degree-corrected stochastic block models}.
\newblock \bibinfo{journal}{Journal of Machine Learning Research}
  \bibinfo{volume}{22}, \bibinfo{pages}{1--63}.
\bibitem[{Malliaros and Vazirgiannis(2013)}]{malliaros2013clustering}
\bibinfo{author}{Malliaros, F.D.}, \bibinfo{author}{Vazirgiannis, M.},
  \bibinfo{year}{2013}.
\newblock \bibinfo{title}{{Clustering and community detection in directed
  networks: A survey}}.
\newblock \bibinfo{journal}{Physics Reports} \bibinfo{volume}{533},
  \bibinfo{pages}{95--142}.
\bibitem[{McDaid et~al.(2013)McDaid, Murphy, Friel and
  Hurley}]{mcdaid2013improved}
\bibinfo{author}{McDaid, A.F.}, \bibinfo{author}{Murphy, T.B.},
  \bibinfo{author}{Friel, N.}, \bibinfo{author}{Hurley, N.J.},
  \bibinfo{year}{2013}.
\newblock \bibinfo{title}{Improved bayesian inference for the stochastic block
  model with application to large networks}.
\newblock \bibinfo{journal}{Computational Statistics \& Data Analysis}
  \bibinfo{volume}{60}, \bibinfo{pages}{12--31}.
\bibitem[{Opsahl et~al.(2010)Opsahl, Agneessens and Skvoretz}]{opsahl2010node}
\bibinfo{author}{Opsahl, T.}, \bibinfo{author}{Agneessens, F.},
  \bibinfo{author}{Skvoretz, J.}, \bibinfo{year}{2010}.
\newblock \bibinfo{title}{{Node centrality in weighted networks: Generalizing
  degree and shortest paths}}.
\newblock \bibinfo{journal}{Social Networks} \bibinfo{volume}{32},
  \bibinfo{pages}{245--251}.
\bibitem[{Qing(2025)}]{qing2025discovering}
\bibinfo{author}{Qing, H.}, \bibinfo{year}{2025}.
\newblock \bibinfo{title}{Discovering overlapping communities in multi-layer
  directed networks}.
\newblock \bibinfo{journal}{Chaos, Solitons \& Fractals} \bibinfo{volume}{194},
  \bibinfo{pages}{116175}.
\bibitem[{Qing and Wang(2023)}]{qing2023community}
\bibinfo{author}{Qing, H.}, \bibinfo{author}{Wang, J.}, \bibinfo{year}{2023}.
\newblock \bibinfo{title}{Community detection for weighted bipartite networks}.
\newblock \bibinfo{journal}{Knowledge-Based Systems} \bibinfo{volume}{274},
  \bibinfo{pages}{110643}.
\bibitem[{Rohe et~al.(2016)Rohe, Qin and Yu}]{rohe2016co}
\bibinfo{author}{Rohe, K.}, \bibinfo{author}{Qin, T.}, \bibinfo{author}{Yu,
  B.}, \bibinfo{year}{2016}.
\newblock \bibinfo{title}{Co-clustering directed graphs to discover asymmetries
  and directional communities}.
\newblock \bibinfo{journal}{Proceedings of the National Academy of Sciences}
  \bibinfo{volume}{113}, \bibinfo{pages}{12679--12684}.
\bibitem[{Saldana et~al.(2017)Saldana, Yu and Feng}]{saldana2017many}
\bibinfo{author}{Saldana, D.F.}, \bibinfo{author}{Yu, Y.},
  \bibinfo{author}{Feng, Y.}, \bibinfo{year}{2017}.
\newblock \bibinfo{title}{How many communities are there?}
\newblock \bibinfo{journal}{Journal of Computational and Graphical Statistics}
  \bibinfo{volume}{26}, \bibinfo{pages}{171--181}.
\bibitem[{Su et~al.(2024)Su, Guo, Chang and Yang}]{su2024spectral}
\bibinfo{author}{Su, W.}, \bibinfo{author}{Guo, X.}, \bibinfo{author}{Chang,
  X.}, \bibinfo{author}{Yang, Y.}, \bibinfo{year}{2024}.
\newblock \bibinfo{title}{Spectral co-clustering in multi-layer directed
  networks}.
\newblock \bibinfo{journal}{Computational Statistics \& Data Analysis}
  \bibinfo{volume}{198}, \bibinfo{pages}{107987}.
\bibitem[{Tropp(2012)}]{tropp2012user}
\bibinfo{author}{Tropp, J.A.}, \bibinfo{year}{2012}.
\newblock \bibinfo{title}{User-friendly tail bounds for sums of random
  matrices}.
\newblock \bibinfo{journal}{Foundations of Computational Mathematics}
  \bibinfo{volume}{12}, \bibinfo{pages}{389--434}.
\bibitem[{Wang and Bickel(2017)}]{LikeAos2017}
\bibinfo{author}{Wang, Y.X.R.}, \bibinfo{author}{Bickel, P.J.},
  \bibinfo{year}{2017}.
\newblock \bibinfo{title}{{Likelihood-based model selection for stochastic
  block models}}.
\newblock \bibinfo{journal}{Annals of Statistics} \bibinfo{volume}{45},
  \bibinfo{pages}{500 -- 528}.
\bibitem[{Wang et~al.(2020)Wang, Liang and Ji}]{wang2020spectral}
\bibinfo{author}{Wang, Z.}, \bibinfo{author}{Liang, Y.}, \bibinfo{author}{Ji,
  P.}, \bibinfo{year}{2020}.
\newblock \bibinfo{title}{Spectral algorithms for community detection in
  directed networks}.
\newblock \bibinfo{journal}{Journal of Machine Learning Research}
  \bibinfo{volume}{21}, \bibinfo{pages}{1--45}.
\bibitem[{Wu and Hu(2024a)}]{wu2024spectral}
\bibinfo{author}{Wu, Q.}, \bibinfo{author}{Hu, J.}, \bibinfo{year}{2024}a.
\newblock \bibinfo{title}{A spectral based goodness-of-fit test for stochastic
  block models}.
\newblock \bibinfo{journal}{Statistics \& Probability Letters}
  \bibinfo{volume}{209}, \bibinfo{pages}{110104}.
\bibitem[{Wu and Hu(2024b)}]{wu2024two}
\bibinfo{author}{Wu, Q.}, \bibinfo{author}{Hu, J.}, \bibinfo{year}{2024}b.
\newblock \bibinfo{title}{Two-sample test of stochastic block models}.
\newblock \bibinfo{journal}{Computational Statistics \& Data Analysis}
  \bibinfo{volume}{192}, \bibinfo{pages}{107903}.
\bibitem[{Zhang and Wang(2022)}]{zhang2022identifiability}
\bibinfo{author}{Zhang, J.}, \bibinfo{author}{Wang, J.}, \bibinfo{year}{2022}.
\newblock \bibinfo{title}{Identifiability and parameter estimation of the
  overlapped stochastic co-block model}.
\newblock \bibinfo{journal}{Statistics and Computing} \bibinfo{volume}{32},
  \bibinfo{pages}{57}.
\bibitem[{Zhou and Amini(2019)}]{zhou2019analysis}
\bibinfo{author}{Zhou, Z.}, \bibinfo{author}{Amini, A.A.},
  \bibinfo{year}{2019}.
\newblock \bibinfo{title}{Analysis of spectral clustering algorithms for
  community detection: the general bipartite setting}.
\newblock \bibinfo{journal}{Journal of Machine Learning Research}
  \bibinfo{volume}{20}, \bibinfo{pages}{1--47}.
\bibitem[{Zhou and Amini(2020)}]{zhou2020optimal}
\bibinfo{author}{Zhou, Z.}, \bibinfo{author}{Amini, A.A.},
  \bibinfo{year}{2020}.
\newblock \bibinfo{title}{Optimal bipartite network clustering}.
\newblock \bibinfo{journal}{Journal of Machine Learning Research}
  \bibinfo{volume}{21}, \bibinfo{pages}{1--68}.

\end{thebibliography}
\end{document}